\documentclass[a4paper]{article}
\usepackage[utf8]{inputenc}

\title{Testing the Missing Completely at Random Assumption for Functional Data}

\usepackage{authblk}
\usepackage{amsmath, amsfonts, amssymb, amsthm}
\usepackage[english]{babel}
\usepackage{bm}
\usepackage{bbm}
\usepackage{enumitem}
\usepackage{framed}
\usepackage{geometry}
\usepackage{graphicx}
\usepackage{subcaption}
\usepackage{hyperref}
\usepackage{cleveref}
\usepackage{mathtools}
\usepackage{multirow}
\usepackage{nicematrix}
\usepackage{lipsum}
\usepackage{soul}
\usepackage[ruled,noend]{algorithm2e}
\usepackage{natbib}
\usepackage{multibib}
\usepackage{color}
\usepackage{soul}
\usepackage{tikz}
\usepackage{floatpag}
\usepackage{bibunits}
\defaultbibliographystyle{apalike}

\DeclareMathOperator*{\argmin}{arg\,min}
\DeclarePairedDelimiter\abs{\lvert}{\rvert}
\newcommand{\ind}{\mathbbm{1}}
\newcommand{\norm}[1]{\lVert#1\rVert}

\newcommand{\var}{\textup{Var}}
\newcommand{\expv}{E}
\newcommand{\prob}{P}
\newcommand{\cov}{\mathbb{C}\textup{ov}}

\newcommand{\hatFA}{\widehat{F}_{\widehat{A}}}
\newcommand{\hatFB}{\widehat{F}_{\widehat{B}}}
\newcommand{\FA}{F_A}
\newcommand{\FB}{F_B}
\newcommand{\hatmA}{\widehat{\mu}_{\widehat{A}}}
\newcommand{\hatmB}{\widehat{\mu}_{\widehat{B}}}
\newcommand{\mA}{\mu_A}
\newcommand{\mB}{\mu_B}

\newcommand{\dt}{\textup{d}t}

\newcommand{\dF}{\textup{d}(\lambda \otimes \nu)}
\newcommand{\dnu}{\textup{d}\nu}
\newcommand{\sym}{\,\triangle\,}

\theoremstyle{definition}

\theoremstyle{plain}
\newtheorem{theorem}{Theorem}
\newtheorem{remark}{Remark}
\newtheorem{proposition}{Proposition}
\newtheorem{lemma}{Lemma}

\newtheorem{assumption}{Assumption}
\newtheorem{example}{Example}

\author[1]{Maximilian Ofner}
\author[1]{Siegfried H\"ormann}
\author[2]{David Kraus}
\author[3]{Dominik Liebl}
\affil[1]{Graz University of Technology, Austria}
\affil[2]{Masaryk University, Czech Republic}
\affil[3]{University Bonn, Germany}

\begin{document}

\maketitle

\begin{abstract}
We consider functional data which have only been observed on a subset of their domain.  This paper aims to develop statistical tests to determine whether the function and the domain over which it is observed are independent. The assumption that data are missing completely at random (MCAR) is essential for many functional data methods handling incomplete observations. However, no general testing procedures have been established to validate this assumption. We address this critical gap by introducing a testing framework which is generally based on a partition of the observation patterns. Besides deterministic partitions, we also consider a systematic approach based on clustering. We establish asymptotic results for our tests and illustrate the methodology in several real data applications.
\end{abstract}

\noindent%

\begin{bibunit}
\section{Introduction}

Over the past decades, functional data analysis (FDA) has garnered significant attention in statistical research and in various fields of applications. FDA provides a statistical framework for analyzing samples of functions varying over a continuous domain. For an introduction to methods in FDA, we refer to the monographs \cite{ramsay1997functional}, \cite{ferraty2006nonparametric}, \cite{horvath2012inference} or \cite{hsing2015theoretical}. Depending on the design and structure of the data, one considers functions which are observed in one of three ways: (1) continuously, (2) on a dense grid, or (3) sparsely. In each case, the measurements are generally supposed to be adequately distributed over a common domain $I\subset \mathbb{R}$.

More recently, efforts have been directed toward developing techniques for functional data samples with varying domains, where each function is observable only on a subset $O \subset I$ and missing on $I \setminus O$. Corresponding data has often been termed \textit{partially observed functional data}, \textit{functional fragments} or \textit{functional snippets}. An early example is \cite{kraus2015components}, where incomplete heart rate profiles are analyzed. A subsample of the original data is depicted in Figure~\ref{fig:heartrate} (left); see also Section~\ref{sec:HeartRateData} for more information. The plot shows heart rate profiles (values in beats per minute) of~$n = 50$ participants during the time interval 8~pm and 2~am. While some profiles were measured on the full domain, several profiles are only partially observed due to a removal or failure of the measuring device.

\begin{figure}
\centering
\includegraphics[width=1\linewidth]{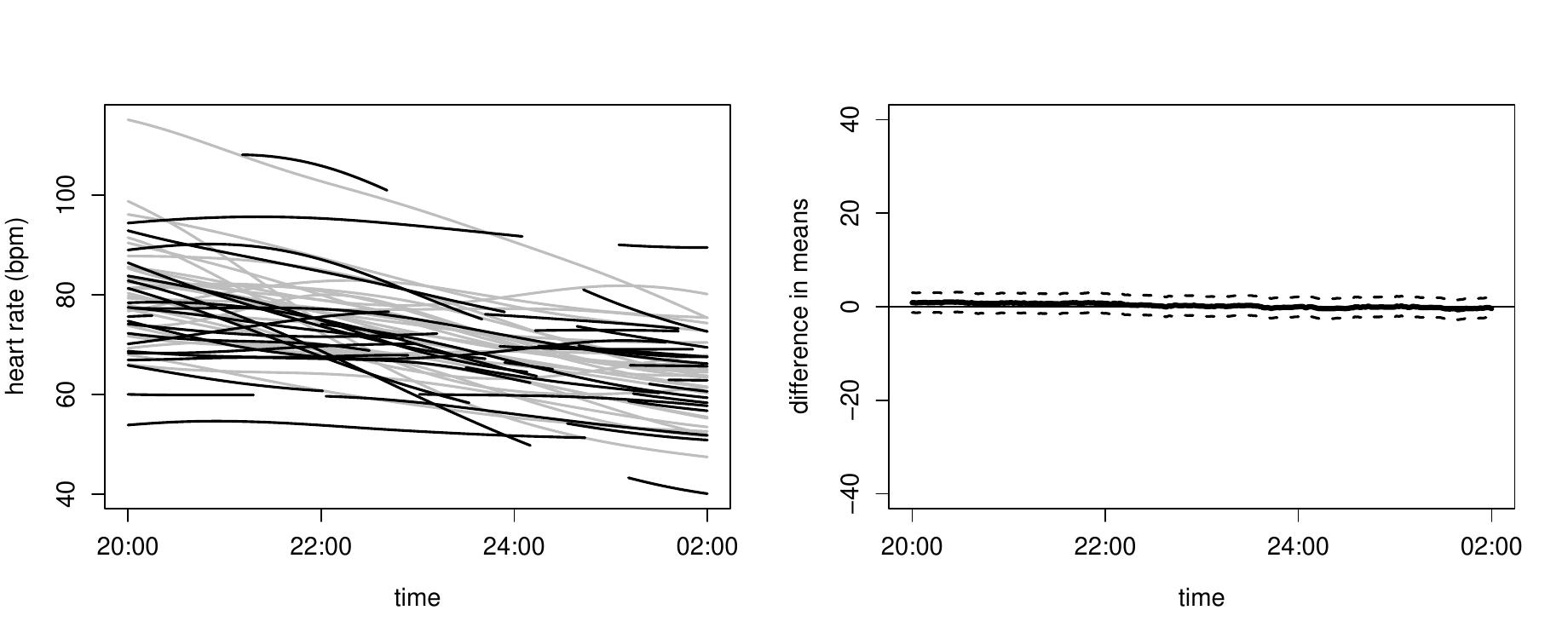}
\caption{Subsample of $n = 50$ heart rate profiles with complete functions colored gray and incomplete functions colored black (left). The conditional means of complete and incomplete groups estimated over 878 subjects do not reveal a significant difference; the dashed lines refer to a 95\% simultaneous confidence band (right).}
\label{fig:heartrate}
\end{figure}

The foremost research question which has been addressed to date for the above and related data sets is how to recover the missing information. Several reconstruction procedures have been proposed in the literature {\citep[see][]{delaigle2016approximating,descary2019recovering,kneip2020on, kraus2020ridge, delaigle2021estimating,lin2022mean,elias2023depth, ofner2024covariateinformed}}.  In this work, we focus on another crucial question which has not yet received much attention: what is the missingness mechanism, and how does it depend on the observed data? Specifically, we propose several tests to assess the hypothesis that the missingness mechanism is independent of the data, commonly known as \textit{missing completely at random} (MCAR).  The importance of this hypothesis is underlined by the fact that all of the above references and many related papers are crucially based on this assumption \citep[see][among many others]{liebl2013modeling,delaigle2013classification,liebl2019nonparametric,kraus2019inferential, kraus2019classification,lin2021basis, elias2023integrated, elias2024statistical}. As we are going to demonstrate, our methodology provides several new insights into the analysis of real data:
\begin{itemize}
\item in an example of heart rate profiles, our tests do not reveal evidence against the MCAR assumption which confirms the opinion of epidemiologists;
\item in an example of electricity price curves, our tests detect clear violations of the MCAR assumption which is in line with the structure of the electricity market;
\item in an example of temperature data, our tests reveal a potential defect of the temperature sensor and thus help to locate the technical cause of missingness.
\end{itemize}
Further details are discussed in Sections~\ref{sec:HeartRateData}-\ref{sec:TemperatureData} and Figure~\ref{fig:examples}.

To the best of our knowledge, the current literature on MCAR testing in the context of FDA is limited to \cite{liebl2019partially}. However, their approach only considers specific violations of the MCAR assumption that can be expressed using a low-dimensional monomial basis representation, enabling the use of a classical multivariate testing method. In contrast, our work explores a much broader range of MCAR violations, extending beyond those in \cite{liebl2019partially}. Our approach accounts for violations that affect the entire distribution of the functional data, not just its low-dimensional components.

To explain our core idea, consider again the heart rate data in Figure~\ref{fig:heartrate}. If the data were MCAR, the distribution of $X$ does not depend on $O$ and the curves in the complete and incomplete group should thus reveal similar characteristics. Indeed, Figure~\ref{fig:heartrate} (right) shows that the groups cannot be distinguished by their means, which provides evidence for the MCAR assumption. More generally, we describe our testing methodology by the following principle:
\begin{enumerate}
\item Assign each curve to one of two groups. This assignment is determined solely by $O$. An intuitive example is the ad hoc partitioning of data into complete and incomplete functions (see also Figure~\ref{fig:heartrate}), though a systematic clustering approach will also be discussed.
\item Compare specific characteristics of the curves between these two groups. Under MCAR, both subsamples should share the same characteristics, as the partition is independent of the data.
\end{enumerate}

The contributions of the underlying work are as follows. We propose a novel framework for testing MCAR based on a partition of the observation sets and discuss a systematic procedure for the selection of the corresponding groups. Then, we propose two different test statistics for the comparison of the groups and derive the asymptotic distributions. To this end, we also establish a theoretical framework for the analysis of càdlàg-valued (discontinuous) functional data that lie in the Skorokhod space $D[0,1]$ instead of the commonly used $L^2[0,1]$ or $C[0,1]$ function spaces. Such a viewpoint is useful when studying uniform properties of empirical mean functions which naturally contain jumps in the case of partial observations. In addition, we introduce a method for constructing simultaneous confidence bands for the difference in mean functions of the groups discussed above. These bands provide an attractive graphical tool for detecting potential violations of MCAR. Moreover, we comment on the consistency of our proposed tests under general alternatives. The testing framework is finally illustrated in a simulation study and three real data examples.

\section{Setup and notation}

Let $X = (X(t)\colon t \in [0,1])$ be a stochastic process which is defined on a probability space $(\Omega, \mathcal{F}, \mathbb{P})$. The domain is usually interpreted as a (scaled) time interval and we refer to $X(t)$ as the value of the process $X$ at time $t \in [0,1]$. Define $\mathcal{F} = \sigma(X(t)\colon t \in [0,1])$ to be the $\sigma$-algebra generated by the complete process and set $\mathcal{F}_S = \sigma(X(t)\colon t \in S)$ to be the $\sigma$-algebra generated by the information of $X$ on a measurable set $S \subset [0,1]$. Let $\mathcal{S}$ denote the space of measurable subsets of $[0,1]$. We henceforth assume that the realized sample paths lie in the space $D[0,1]$ of càdlàg functions, that is, functions which are right-continuous with left limits. Equipped with the Skorokhod metric, this space is Polish. As pointed out by \cite[p.~135]{billingsley2013convergence}, the process then defines a random function $X\colon \Omega \to D[0,1]$. Assume further that $X$ is only observed on a subset $O$ of its domain. This subset may vary from curve to curve and is thus interpreted as random set which takes values in the space $\mathcal{K}$ of non-empty compact subsets of $[0,1]$. Equipped with the Hausdorff metric, $\mathcal{K}$ is a compact metric space and therefore Polish \citep{molchanov2017theory}. Associated with~$O$ is the \textit{indicator function} $\xi_O\colon t \mapsto \ind\{t \in O\}$ which can be regarded as a random function.  Note that the information in our case is twofold as it includes both the processes $X$ as well as the indicator function~$\xi_O$. Whereas $X$ is not fully available in the missing data case, we still observe the point-wise product $X\xi_O = (X(t)\xi_O(t)\colon t \in [0,1])$ and the indicator function~$\xi_O$. The information of the indicator function $\xi_O(t)$ at time $t$ allows us to distinguish the case that~$X(t)$ is unobserved from $X(t)$ being equal to zero. An illustration is given in Figure~\ref{fig:illustration}.

Throughout this work, let $(X_i, O_i)_{i=1}^n$ be $n$ i.i.d.~copies of the generic $(X,O)$. It is important to note that $O_i$ may very well depend on $X_i$ (in the case that MCAR is violated).

\begin{figure}
\centering
\begin{tikzpicture}[x=1pt,y=1pt]
\definecolor{fillColor}{RGB}{255,255,255}
\path[use as bounding box,fill=fillColor,fill opacity=0.00] (0,0) rectangle (289.08, 72.27);
\begin{scope}
\path[clip] ( 50.40,  2.40) rectangle (286.68, 69.87);
\definecolor{drawColor}{RGB}{0,0,0}

\path[draw=drawColor,line width= 1.6pt,line join=round,line cap=round] ( 59.15, 24.07) --
	( 61.34, 25.80) --
	( 63.53, 27.52) --
	( 65.71, 29.22) --
	( 67.90, 30.88) --
	( 70.09, 32.49) --
	( 72.28, 34.04) --
	( 74.47, 35.53) --
	( 76.65, 36.93) --
	( 78.84, 38.24) --
	( 81.03, 39.46) --
	( 83.22, 40.58) --
	( 85.40, 41.59) --
	( 87.59, 42.50) --
	( 89.78, 43.29) --
	( 91.97, 43.97) --
	( 94.16, 44.54) --
	( 96.34, 45.00) --
	( 98.53, 45.36) --
	(100.72, 45.62) --
	(102.91, 45.79) --
	(105.09, 45.86) --
	(107.28, 45.86) --
	(109.47, 45.78) --
	(111.66, 45.64) --
	(113.85, 45.44) --
	(116.03, 45.20) --
	(118.22, 44.92) --
	(120.41, 44.62) --
	(122.60, 44.29) --
	(124.78, 43.95) --
	(126.97, 43.61) --
	(129.16, 43.27) --
	(131.35, 42.95) --
	(133.54, 42.63) --
	(135.72, 42.34) --
	(137.91, 42.07) --
	(140.10, 41.83) --
	(142.29, 41.62) --
	(144.47, 41.43) --
	(146.66, 41.26) --
	(148.85, 41.13) --
	(151.04, 41.01) --
	(153.23, 40.91) --
	(155.41, 40.83) --
	(157.60, 40.75) --
	(159.79, 40.68) --
	(161.98, 40.60) --
	(164.16, 40.51) --
	(166.35, 40.40) --
	(168.54, 40.26) --
	(170.73, 40.09) --
	(172.92, 39.88) --
	(175.10, 39.63) --
	(177.29, 39.32) --
	(179.48, 38.96) --
	(181.67, 38.53) --
	(183.85, 38.04) --
	(186.04, 37.48) --
	(188.23, 36.85);
\end{scope}
\begin{scope}
\path[clip] (  0.00,  0.00) rectangle (289.08, 72.27);
\definecolor{drawColor}{RGB}{0,0,0}

\path[draw=drawColor,line width= 0.4pt,line join=round,line cap=round] ( 50.40,  2.40) --
	(286.68,  2.40) --
	(286.68, 69.87) --
	( 50.40, 69.87) --
	cycle;
\end{scope}
\begin{scope}
\path[clip] (  0.00,  0.00) rectangle (289.08, 72.27);
\definecolor{drawColor}{RGB}{0,0,0}

\node[text=drawColor,anchor=base,inner sep=0pt, outer sep=0pt, scale=  1.00] at (168.54,-31.20) {Index};

\node[text=drawColor,rotate= 90.00,anchor=base,inner sep=0pt, outer sep=0pt, scale=  1.00] at ( 24.00, 36.13) {$X(t)$};
\end{scope}
\begin{scope}
\path[clip] ( 50.40,  2.40) rectangle (286.68, 69.87);
\definecolor{drawColor}{RGB}{190,190,190}

\node[text=drawColor,anchor=base,inner sep=0pt, outer sep=0pt, scale=  3.00] at (236.36, 16.36) {?};
\end{scope}
\end{tikzpicture}
\begin{tikzpicture}[x=1pt,y=1pt]
\definecolor{fillColor}{RGB}{255,255,255}
\path[use as bounding box,fill=fillColor,fill opacity=0.00] (0,0) rectangle (289.08, 72.27);
\begin{scope}
\path[clip] ( 48.00, 50.40) rectangle (286.68, 69.87);
\definecolor{drawColor}{RGB}{0,0,0}

\path[draw=drawColor,line width= 0.8pt,line join=round,line cap=round] ( 56.84, 69.15) --
	(189.44, 69.15);
\end{scope}
\begin{scope}
\path[clip] (  0.00,  0.00) rectangle (289.08, 72.27);
\definecolor{drawColor}{RGB}{0,0,0}

\path[draw=drawColor,line width= 0.4pt,line join=round,line cap=round] ( 56.84, 50.40) -- (277.84, 50.40);

\path[draw=drawColor,line width= 0.4pt,line join=round,line cap=round] ( 56.84, 50.40) -- ( 56.84, 44.40);

\path[draw=drawColor,line width= 0.4pt,line join=round,line cap=round] (101.04, 50.40) -- (101.04, 44.40);

\path[draw=drawColor,line width= 0.4pt,line join=round,line cap=round] (145.24, 50.40) -- (145.24, 44.40);

\path[draw=drawColor,line width= 0.4pt,line join=round,line cap=round] (189.44, 50.40) -- (189.44, 44.40);

\path[draw=drawColor,line width= 0.4pt,line join=round,line cap=round] (233.64, 50.40) -- (233.64, 44.40);

\path[draw=drawColor,line width= 0.4pt,line join=round,line cap=round] (277.84, 50.40) -- (277.84, 44.40);

\node[text=drawColor,anchor=base,inner sep=0pt, outer sep=0pt, scale=  1.00] at ( 56.84, 34.80) {0.0};

\node[text=drawColor,anchor=base,inner sep=0pt, outer sep=0pt, scale=  1.00] at (101.04, 34.80) {0.2};

\node[text=drawColor,anchor=base,inner sep=0pt, outer sep=0pt, scale=  1.00] at (145.24, 34.80) {0.4};

\node[text=drawColor,anchor=base,inner sep=0pt, outer sep=0pt, scale=  1.00] at (189.44, 34.80) {0.6};

\node[text=drawColor,anchor=base,inner sep=0pt, outer sep=0pt, scale=  1.00] at (233.64, 34.80) {0.8};

\node[text=drawColor,anchor=base,inner sep=0pt, outer sep=0pt, scale=  1.00] at (277.84, 34.80) {1.0};
\end{scope}
\begin{scope}
\path[clip] (  0.00,  0.00) rectangle (289.08, 72.27);
\definecolor{drawColor}{RGB}{0,0,0}

\node[text=drawColor,anchor=base,inner sep=0pt, outer sep=0pt, scale=  1.00] at (167.34, 16.80) {$t$};

\node[text=drawColor,rotate= 90.00,anchor=base,inner sep=0pt, outer sep=0pt, scale=  1.00] at ( 21.60, 60.13) {$\xi_O(t)$};
\end{scope}
\begin{scope}
\path[clip] (  0.00,  0.00) rectangle (289.08, 72.27);
\definecolor{drawColor}{RGB}{0,0,0}

\node[text=drawColor,rotate= 90.00,anchor=base,inner sep=0pt, outer sep=0pt, scale=  1.00] at ( 39.60, 51.12) {0};

\node[text=drawColor,rotate= 90.00,anchor=base,inner sep=0pt, outer sep=0pt, scale=  1.00] at ( 39.60, 69.15) {1};
\end{scope}
\begin{scope}
\path[clip] ( 48.00, 50.40) rectangle (286.68, 69.87);
\definecolor{drawColor}{RGB}{0,0,0}

\path[draw=drawColor,line width= 0.8pt,line join=round,line cap=round] (189.44, 51.12) --
	(277.84, 51.12);
\end{scope}
\end{tikzpicture}
\vspace{-0.5cm}
\caption{Illustration of random function $X$ and indicator function $\xi_O$ for $O = [0,0.6]$.}
\label{fig:illustration}
\end{figure}

\subsection{Hypothesis}

We consider the \textit{missing completely at random} assumption
\begin{equation}\label{MCAR}
H_0\colon \text{$X$ and $O$ are independent.}\tag{MCAR}
\end{equation}
This assumption concerns a probabilistic property of the underlying missingness mechanism and is equivalent to $\prob(O \subset K \vert \mathcal{F}) = \prob(O \subset K)$ for all $K \in \mathcal{K}$. The quantity $\prob(O \subset K)$ is called \textit{containment functional} and uniquely determines the distribution of $O$ \citep[p.~23]{molchanov2017theory}. A weaker condition is \textit{missing at random} (MAR), which can be defined similarly to \cite{farewell2022missing} as
\begin{equation}
\prob\left(O \subset K \vert \mathcal{F}\right) = \prob\left(O \subset K \vert \mathcal{F}_K\right), \qquad \text{for all $K  \in \mathcal{K}$.}\tag{MAR}
\end{equation}
Whereas MCAR is a statement about independence of two random variables, MAR should not be confused with conditional independence between random variables. Related discussions can be found in \cite{seaman2013meant}, \cite{mealli2015clarifying}, and \cite{farewell2022missing}.  Historically, a first thorough treatment of missing data was given by the seminal work of \cite{rubin1976inference}. Data that are neither MCAR nor MAR are called \textit{missing not at random} (MNAR). We will discuss different examples in the simulation study in Section~\ref{sec:NumericalIllustration}.

Aside from the fact that we want to test independence between two functional objects, which is a delicate problem on its own \citep[see, for instance][]{deb2020measuring}, the crucial difficulty we face in this context is that $X$ is not observable; we only observe $X$ on $O$. Hence, existing approaches for complete data are not applicable here. For multivariate data with missing values, the arguably most prominent test for MCAR is due to \cite{little1988test}. Recent contributions to this topic include \cite{zhangetal2019unified}, \cite{berrett2023optimal}, \cite{aleksic2024novel}, and \cite{bordino2024tests}. In the case of $\mathbb{R}^d$, the number of distinct observation patterns is bounded by~$2^d$. It is therefore possible to group data according to these patterns and consider an asymptotic scenario where the number of occurrences of each missingness pattern tends to infinity.
The main idea is then to verify whether the subdistributions estimated from these observation patterns are coherent.
This approach is generally infeasible for functional data since, in practice, all observation patterns are distinct. At the population level, we even have an uncountably infinite number of distinct observation sets. Hence, we propose to partition the observation patterns into two systematically chosen sets $\mathcal{A}$ and $\mathcal{B}$ and test if the distribution of $X$ depends on whether $O\in\mathcal{A}$ or $O\in\mathcal{B}$.

\subsection{Testing framework}

The random set $O$ takes values in the space $\mathcal{K}$ of non-empty compact subsets of $[0,1]$. We then consider a partition
$\mathcal{K} = \mathcal{A} \mathbin{\dot{\cup}} \mathcal{B}$ for given non-empty subsets $\mathcal{A}$ and $\mathcal{B}$, where $\mathcal{A} \mathbin{\dot{\cup}} \mathcal{B}=\mathcal{A} \cup \mathcal{B}$ with $\mathcal{A} \cap \mathcal{B}=\emptyset$ denotes the disjoint set union. A straightforward approach is to divide the data into complete and incomplete cases, or, more generally, to classify observations based on whether the Lebesgue measure of their observation set falls above or below a chosen threshold. A systematic approach based on clustering is discussed in Section~$\ref{sec:Clustering}$. We impose the following general assumptions on the partition.
\begin{assumption}\label{ass:Missing}
Let $A = \{O \in \mathcal{A}\}$ and $B = \{O \in \mathcal{B}\}$. The partition satisfies
\begin{enumerate}[label=(\roman*),topsep=0pt, partopsep=0pt]
\item $\prob(A) > 0$ and $\prob(B) > 0$,
\item $\inf_{t \in [0,1]} \{\prob(t \in O \vert A), \prob(t \in O \vert B)\} > 0$.
\end{enumerate}
\end{assumption}
Condition~(i) excludes the trivial case where all $O$ are attributed to the same subset. Condition~(ii) considers group-specific \textit{coverage functions} $t \mapsto \prob(t \in O\vert \cdot)$ and assumes that the information in both subsamples covers the domain $[0,1]$.

\begin{remark}\label{rem:I}
    Assumption~\ref{ass:Missing}~(ii) is not always satisfied in practice. In such a case, one could restrict the analysis to a closed set $I \subset [0,1]$ such that $\inf_{t \in I} \{\prob(t \in O \vert A), \prob(t \in O \vert B)\} > 0$. Proposition~\ref{prop:ChoiceOfI} in Section~\ref{se:centroid} discusses a generic choice of $I$ which guarantees (ii) for a systematic partition.
\end{remark}

Now since $\mathbb{R} \times \mathcal{K}$ is a Polish space, the regular conditional distributions $\prob(X(t) \leq z \vert A)$ and $\prob(X(t) \leq z \vert B)$ exist \citep[see, for example, Theorem~10.2.2 in][]{dudley2002real}. For our testing problem, we then focus on directional alternatives which violate
\begin{equation*}
    H_0^F \colon\quad\prob(X(t) \leq z \vert A) = \prob(X(t) \leq z \vert B), \qquad \text{for all $z \in \mathbb{R}, t \in [0,1]$},
\end{equation*}
where $H_0^F\supset H_0$. Hypothesis $H_0^F$ states that, for all $t \in [0,1]$, the conditional distribution of~$X(t)$ does not depend on whether $\{O \in \mathcal{A}\}$ or $\{O \in \mathcal{B}\}$. This hypothesis is clearly satisfied under MCAR, where~$X(t)$ is independent of $O$ and the violation of $H_0^F$ covers many practically important MCAR violations (see Section~\ref{sec:Consistency}). A further comparison of the hypotheses $H_0$ and~$H_0^F$ is given in the supplementary material. Note that it is in general impossible to
make inference about $\prob(X(t) \leq z \vert A)$ and $\prob(X(t)\leq z \vert B)$ as these terms are not estimable due
to the missing data. However, if MCAR holds true, these quantities are equal to $\prob(X(t) \leq z \vert t\in O,A)$ and $\prob(X(t) \leq z \vert t \in O, B)$, respectively, which are estimable from the available data. The general steps of our testing framework are summarized in Algorithm~\ref{alg:DescriptionMCAR} below.

\begin{algorithm}[H]
\caption{Description of MCAR testing framework}\label{alg:DescriptionMCAR}
\nl Choose a partition $\mathcal{K} = \mathcal{A} \mathbin{\dot{\cup}} \mathcal{B}$\;
\nl Reject MCAR if 
\begin{equation*}
H_0^F \colon\quad \prob(X(t) \leq z \vert O \in \mathcal{A}) = \prob(X(t) \leq z \vert O \in \mathcal{B}), \qquad \text{for all $z \in \mathbb{R}$, $t \in [0,1],$}
\end{equation*}
is violated\;
\end{algorithm}

\begin{remark}
Two questions naturally arise: why do we partition the observation patterns into only two groups, and why do we restrict the null hypothesis $H_0^F$ to marginal distributions? The main motivation in both cases is the limited coverage of the partially observed data and, in particular, the resulting constraints on practical feasibility. Introducing additional groups would further reduce the available joint observation domain (see Remark~\ref{rem:I}), thereby exacerbating data sparsity. This effect is particularly pronounced for finite-dimensional distributions (compare the bivariate case, where the joint domain is a subset of $[0,1]^2$), resulting in an even more substantial loss of information.
\end{remark}

\section{Clustering of observation sets}\label{sec:Clustering}

As outlined in Step~1 of Algorithm~\ref{alg:DescriptionMCAR}, our testing methodology requires a partition $\mathcal{K} = \mathcal{A} \mathbin{\dot{\cup}} \mathcal{B}$ of the observation sets. In this section, we propose a systematic procedure for the choice of this partition. The basic idea is to divide the observation sets into two clusters such that sets within a single cluster are more closely related than sets from different clusters. If a dependence exists between observation sets and functions, then similar (or dissimilar) observation sets are expected to correspond to similar (or dissimilar) functions. In this case, functions from different clusters will exhibit distinct characteristics and MCAR is rejected. While our methodology can be combined with any clustering algorithm, we restrict ourselves to a centroid-based clustering approach, which aligns well with our theoretical framework.

\subsection{Distance between sets}

To discuss a clustering procedure for set-valued random variables, we introduce a distance first. For two sets $S_1, S_2 \in \mathcal{S}$, define $S_1 \sym  S_2 = (S_1 \setminus S_2) \cup (S_2 \setminus S_1)$ to be the symmetric difference. Let~$\lambda$ denote the Lebesgue measure on $[0,1]$ and define $d(S_1,S_2) = \lambda(S_1 \sym S_2)$ which is sometimes referred to as \textit{Fréchet-Nikodym-Aronszajn} distance \citep{marczewski1958on}. This yields a Polish space $\mathcal{M} = \mathcal{S} / {\sim}$ that consists of equivalence classes of sets in $\mathcal{S}$ where $S_1 \sim S_2$ if and only if $\lambda(S_1\sym S_2)=0$. It is important to note that $S_1\sym S_2 = ([0,1]\setminus S_1)\sym ([0,1]\setminus S_2)$ which implies that any clustering based on $d$ will yield equivalent results, regardless of whether the clustering is applied to observation or missingness patterns. This is a desirable property which does not hold if we chose $d$ to be the Hausdorff distance. Moreover, we observe that
\begin{equation*}
d(S_1, S_2) = \int_0^1   \abs{\xi_{S_1}(t)-\xi_{S_2}(t)} \, \dt  = \int_0^1   \abs{\xi_{S_1}(t)-\xi_{S_2}(t)}^2 \, \dt,
\end{equation*}
and hence the Fréchet-Nikodym-Aronszajn distance between sets corresponds to an $L^1$-distance or squared $L^2$-distance of indicator functions \citep{marczewski1958on}. Different notions of centers exist for set-valued random variables and we consider a particular definition which is due to Vorob'ev. The \textit{Vorob'ev median} $M_0$ of $O$ is defined as
$M_0 = \left\lbrace t \in [0,1]\colon \prob(t \in O)\geq 1/2\right\rbrace$
and can be simply computed from the coverage function $t \mapsto \prob(t \in O)$. It satisfies $\expv[d(O,M_0)] \leq \expv[d(O,S)]$ for any measurable set $S \in \mathcal{S}$. For a general introduction to set-valued random variables, the interested reader is referred to \cite{molchanov2017theory}.

\subsection{Centroid-based clustering}\label{se:centroid}

Using the metric $d(S_1,S_2) = \lambda(S_1 \sym S_2)$, we can define two clusters by
\begin{align}\label{eq:Clustering}
\mathcal{A} = \left\lbrace O \in \mathcal{K}\colon d(O, M_{A}) \leq d(O, M_{B}) \right\rbrace, && \mathcal{B} = \left\lbrace O \in \mathcal{K}\colon d(O, M_{A}) > d(O, M_{B}) \right\rbrace,
\end{align}
where $M_A$ and $M_B$ are fixed choices of set-valued centers. For a systematic choice, consider a subspace $\mathcal{M}_0 \subset \mathcal{M}$ and define
\begin{equation}\label{eq:OptimalCenters}
    \{M_A, M_B\} = \argmin_{\{M_1, M_2\} \subset \mathcal{M}_0} \expv\left[ \min\{d(O, M_1), d(O, M_2)\}\right],
\end{equation}
whose elements are called \textit{optimal centers} of $O$. We assume that the minimizer $\{M_A, M_B\}$ of~\eqref{eq:OptimalCenters} is unique which automatically ensures that $M_A \neq M_B$. The optimal centers motivate a generic choice of the domain $I \subset [0,1]$ that satisfies the condition in Remark~\ref{rem:I}.

\begin{proposition}\label{prop:ChoiceOfI}
Consider a centroid-based clustering which is induced by the optimal centers $M_A$ and $M_B$ defined in \eqref{eq:OptimalCenters} for $\mathcal{M}_0 = \mathcal{M}$. Let $I = M_A \cap M_B \neq \emptyset$. Then it holds that
\begin{equation*}
    \inf_{t \in I} \, \{\prob(t \in O\vert A), \prob(t \in O\vert B)\} \geq 1/2.
\end{equation*}
\end{proposition}
The minimization problem in~\eqref{eq:OptimalCenters} relates to the population version of a 2-medians clustering problem and $M_{A}$ as well as $M_{B}$ define the centers of related Voronoi cells. In simple cases, they can be computed explicitly as the following example demonstrates.

\begin{example}[Monotone missingness pattern]\label{ex:[0,D]}
Let $D$ be a random variable on $[0,1]$ with cdf $F_D$ and define $O = [0,D]$ which realizes in the space $\{[0, d]\colon d \in [0,1]\}$. In this case, the centers are of the form $M_1 = [0,m_1]$ and $M_2 = [0,m_2]$ and it holds that
\begin{equation}\label{eq:[0,D]}
\expv \left[ \min\{d(O, M_1), d(O, M_2)\}\right] = \expv \left[ \min\{\abs{D - m_1}, \abs{D - m_2}\}\right].
\end{equation}
Minimizing the right-hand side of \eqref{eq:[0,D]} in $m_1$ and $m_2$ relates to the quantization of a univariate probability distribution $F_D$. Following~\citet[p.~66]{graf2000foundations}, optimal $m_{1}$ and $m_{2}$ need to satisfy the following conditions,
\begin{align*}
2F_D(m_1) = F_D((m_1+m_2)/2), && 2F_D(m_2) = 1+F_D((m_1+m_2)/2).
\end{align*}
For example, if $D \sim U[0,1]$, we obtain $m_{A} = 1/4$ and $m_{B} = 3/4$ which finally yields the optimal centers $M_A = [0,1/4]$ and $M_B = [0,3/4]$. By~Proposition~\ref{prop:ChoiceOfI}, the choice $I = M_A \cap M_B=[0, 1/4]$ then guarantees $\inf_{t \in I} \{\prob(t \in O \vert A), \prob(t \in O \vert B)\} \geq 1/2$ and hence the condition in Remark~\ref{rem:I} holds.
\end{example}
The partition in \eqref{eq:Clustering} clusters the observation intervals of Example~\ref{ex:[0,D]} according to their length. Such sets also appear in the real data application of Section~\ref{sec:electricity}.

\subsection{Estimation of optimal centers}
\label{se:centroid-est}

The optimal centers are generally unknown and need to be estimated from the data. To this end, define the $M$-estimator
\begin{equation*}
\{\widehat{M}_{A}, \widehat{M}_{B}\} = \argmin_{\{M_1, M_2\} \subset \mathcal{M}_0} \frac{1}{n} \sum_{i=1}^n  \min\{d(O_i, M_1), d(O_i, M_2)\}.
\end{equation*}

We tacitly exclude the case $\widehat{M}_A = \widehat{M}_B$, which occurs in the degenerate case where all $O_i$ are identical. The minimum of the above function can be computed by a conventional $2$-median clustering algorithm using the Vorob'ev median to approximate the centers. Our hypothesis tests are invariant to the labeling of $M_A$ and $M_B$, and since these centers are not separately identifiable, we assume without loss of generality that $d(\widehat{M}_A, M_A) \leq d(\widehat{M}_B, M_A)$. The following proposition yields consistency of the empirical centers in our case. 

\begin{proposition}\label{prop:ConsistencyCenters}
Suppose that $\mathcal{M}_0$ is compact with respect to $d$ and the minimizer $\{M_A, M_B\}$ of~\eqref{eq:OptimalCenters} is unique. Then,
$d(\widehat{M}_{A}, M_{A}) +d(\widehat{M}_{B}, M_{B}) = o_p(1)$ as $n \to \infty$.
\end{proposition}
Consistency of $k$-means in the multivariate case has been established by \cite{pollard1981strong}. Our proposition is related to a work of \cite{jiang2021consistency} who derived the consistency of $k$-medoids in general metric and non-metric spaces. An example that satisfies the conditions of Proposition~\ref{prop:ConsistencyCenters} is given by the observation patterns in Example~\ref{ex:[0,D]}. We finally obtain the data-adaptive partitions
\begin{align*}
\widehat{\mathcal{A}} = \left\lbrace O \in \mathcal{K}\colon d(O, \widehat{M}_{A}) \leq d(O, \widehat{M}_{B}) \right\rbrace, && \widehat{\mathcal{B}} = \left\lbrace O \in \mathcal{K}\colon d(O, \widehat{M}_{A}) > d(O, \widehat{M}_{B}) \right\rbrace.
\end{align*}

\begin{lemma}\label{lem:AmathcalA}
Suppose $\prob(d(O,M_{A}) = d(O,M_{B})) = 0$ and $d(\widehat{M}_{A}, M_{A}) + d(\widehat{M}_{B}, M_{B}) = o_p(1)$. For any $i$, it holds that $\prob(O_i \in \mathcal{A} \sym \widehat{\mathcal{A}}) = \prob(O_i \in \mathcal{B} \sym \widehat{\mathcal{B}}) = o(1)$ as $n \to \infty$.
\end{lemma}
For ease of reference, we introduce the following assumption.
\begin{assumption}\label{ass:hatA}
For any $i$, it holds that $\prob(O_i \in \mathcal{A} \sym \widehat{\mathcal{A}}) = o(1)$ as $n \to \infty$.
\end{assumption}

\section{Testing methodology}\label{sec:TestingMethodology}

We now discuss two  approaches for testing the hypothesis $H_0^F$ in Step 2 of Algorithm~\ref{alg:DescriptionMCAR}. First, we consider a direct comparison of means and second, we examine a more general comparison of one-dimensional distributions.

\subsection{Comparison of means}\label{sec:ComparisonOfMeans}

Let $\mu(t) = \expv[X(t)]$ and consider the hypothesis
\begin{equation*}
   H_0^\mu \colon\quad \expv \left[X(t)\vert A \right] = \expv\left[X(t)\vert  B\right], \quad \text{for all $t \in [0,1]$,}
\end{equation*}
where $H_0^\mu\supset H_0^F\supset H_0$. Hypothesis $H_0^\mu$ states that the mean function of $X$ is independent of whether $\{O \in \mathcal{A}\}$ or $\{O \in \mathcal{B}\}$. This provides a useful benchmark for the more general distribution test of the next section. Under MCAR,~$H_0^\mu$ is clearly satisfied, and we can estimate the unknown quantities by the cross-sectionally available data,
\begin{align*}
&\hatmA(t) = \frac{1}{n}\frac{1}{\widehat{p}_{\widehat{A}}(t)} \sum_{i = 1}^n X_i(t)\ind\{t \in O_i, O_i \in \widehat{\mathcal{A}}\}, && \widehat{p}_{\widehat{A}}(t) = \frac{1}{n} \sum_{i=1}^n \ind\{t \in O_i, O_i \in \widehat{\mathcal{A}}\}, && t \in [0,1],
\end{align*}
and likewise $\hatmB(t)$ and $\widehat{p}_{\widehat{B}}(t)$ for $\mathcal{B}$. Assumption~\ref{ass:Missing}~(ii) guarantees that the available information in both subsamples covers the domain $[0,1]$. Thus, the denominator $\widehat{p}_{\widehat{A}}$ will be nonzero if the sample size is sufficiently large.

It follows under our assumptions that
\begin{align*}
\hatmA(t) \quad \xrightarrow{p} \quad \mA(t) = \expv \left[X(t)\vert t \in O, A \right], && \hatmB(t) \quad \xrightarrow{p} \quad \mB(t) = \expv\left[X(t)\vert t \in O, B \right],
\end{align*}
both pointwise and uniformly, as $n \to \infty$. Under~MCAR, it holds that $\mA(t) = \expv[X(t)]$, as well as $\mB(t) = \expv[X(t)]$, and $\hatmA$ should thus be close to $\hatmB$. This motivates a test statistic of the form
\begin{equation*}
T_{\mu} = \sqrt{n} \, \norm{\hatmA - \hatmB}_\infty,
\end{equation*}
where $\norm{f}_\infty = \sup_{t \in [0,1]} \abs{f(t)}$. It is important to note that the empirical means $\hatmA$ and $\hatmB$ generally contain jumps and we thus utilize a central limit theorem for càdlàg-valued random functions. Whereas the central limit theorem requires only mild conditions in Hilbert spaces, more assumptions are required for the Skorokhod space $D[0,1]$. To apply corresponding results in the proof of Theorem~\ref{thm:CLT_D} formulated below, we introduce additional assumptions. To this end, let $C > 1$ be some absolute constant which may have a different value at each appearance. The following assumption is adapted from \cite{hahn1978central}.

\begin{assumption}\label{ass:Y}
The generic process $Y = (Y(t)\colon t \in [0,1])$ satisfies $\expv[Y^4(t)]< C$ and there exists some $\alpha > 1$ such that for all $0 \leq s \leq t \leq u \leq 1$,
\begin{enumerate}[label=(\roman*)]
\item $\expv[(Y(s)-Y(t))^2] \leq C \abs{t - s}^{\alpha/2}$,
\item $\expv[(Y(s)-Y(t))^2(Y(t)-Y(u))^2] \leq C \abs{u - s}^{\alpha}$.
\end{enumerate}
\end{assumption}

\noindent The Brownian motion and Poisson process satisfy Assumption~\ref{ass:Y}. Indeed, since both processes satisfy condition~(i) with $\alpha = 2$, condition~(ii) then directly results from the property of independent increments. An example which satisfies (i) and (ii) with $\alpha = 4$ is discussed in the following.

\begin{example}
Let $X = (X(t)\colon t \in [0,1])$ be a stochastic process with continuously differentiable sample paths such that $\expv[\norm{X'}^4_\infty] < C$. For any $0 \leq s \leq t \leq u \leq 1$, a pathwise application of the mean value theorem then yields
$\expv\left[ (X(s)-X(t))^2\right] \leq \expv[\norm{X'}_\infty^2] (t-s)^2$,
and $\expv\left[ (X(s)-X(t))^2(X(t)-X(u))^2\right] \leq \expv[\norm{X'}_\infty^4] (u-s)^4$.
\end{example}
The next example illustrates a case of an indicator process $\xi_O = (\xi_O(t)\colon t \in [0,1])$ that satisfies Assumption~\ref{ass:Y}.
\begin{example}\label{ex:O}
Consider independent random variables $U_1, U_2$ on $[0,1]$ with continuously differentiable cdf $F$ and define $O = [\min\{U_1, U_2\}, \max\{U_1, U_2\}]\subset [0,1]$. Then, it can be shown that $\prob(\xi_O(t) \neq \xi_O(u)) = 2(F(u)-F(t))(1+F(t)-F(u)) \leq C (u-t)$ and $\prob(\xi_O(s) \neq \xi_O(t), \xi_O(t) \neq \xi_O(u)) \leq 2 (F(t)-F(s))(F(u)-F(t)) \leq C (u-s)^2$. Hence, Assumption~\ref{ass:Y} holds with $\alpha = 2$. 
\end{example}

Equipped with the above assumptions, we are able to derive the asymptotic distribution of the test statistic based on the sup-norm.

\begin{theorem}\label{thm:CLT_D}
Let Assumptions~\ref{ass:Missing} and~\ref{ass:hatA} be satisfied and suppose that both $(X_i)_{i=1}^n$ and $(\xi_{O_i})_{i=1}^n$ are i.i.d.~$D[0,1]$-valued random variables satisfying Assumption~\ref{ass:Y} with $\expv{[\norm{X}_\infty^2]} < \infty$. Under~MCAR, $$T_\mu  \xrightarrow{d} \sup_{t\in [0,1]} \, {\abs{Z(t)}},$$
as $n \to \infty$, where $Z$ denotes a Gaussian process with continuous sample paths, zero mean, and covariance
\begin{equation}\label{eq:CovZ}
\begin{aligned}
k(s,t) = \prob(A)\, \frac{\expv[X^c(s)X^c(t)\xi_O(s)\xi_O(t)\vert A]}{\prob(s\in O, A)\prob(t\in O, A)} + \prob(B)\, \frac{\expv[X^c(s)X^c(t)\xi_O(s)\xi_O(t)\vert B]}{\prob(s\in O, B)\prob(t\in O, B)},
\end{aligned}
\end{equation}
where $X^c(t) = X(t) - \left(\mA(t)\ind\{A\}+\mB(t)\ind\{B\}\right)$.
\end{theorem}
Theorem~\ref{thm:CLT_D} utilizes a connection between partially observed functions and càdlàg processes. We believe that it is an interesting result in its own right, which has not been established in the literature to the best of our knowledge.

For approximating the distribution of $T_{\mu}$, one needs to estimate the covariance function. To this end, note that under MCAR, it actually holds that
\begin{equation}\label{eq:CovFac}
\begin{aligned}
&k(s,t)= \cov(X(s),X(t)) \left(\frac{\prob(s\in O, t\in O,A)}{\prob(s\in O, A)\prob(t\in O, A)} + \frac{\prob(s\in O, t\in O,B)}{\prob(s\in O, B)\prob(t\in O, B)}\right).
\end{aligned}
\end{equation}
Consequently, estimators of $k$ could be constructed from separate estimators for the factors involving~$X$ and $O$, respectively. In this case, it actually suffices to estimate $\cov(X(s), X(t))$ at pairs $(s,t)$ for which $\prob(s\in O, t\in O,A)+ \prob(s\in O, t\in O,B) > 0$; see also Theorem~1 in \cite{kraus2019inferential} for a discussion of covariance estimation in case of partially observed data. However, it is important to remark that the choice between \eqref{eq:CovZ} and \eqref{eq:CovFac} will ultimately affect the power of the resulting test. In the supplementary file, we define an estimator based on~\eqref{eq:CovZ} which consistently estimates the covariance of $\sqrt{n}\,(\hatmA - \hatmB)$ even if MCAR is violated.

Algorithm~\ref{alg:D} in the supplementary file summarizes the required steps for approximating the asymptotic distribution. A slight adaptation actually yields a construction method for simultaneous confidence bands of mean differences~$\hatmA-\hatmB$. This is shown in Algorithm~\ref{alg:SCB} in the supplementary file where a confidence band of constant width is constructed. The procedure could be extended to more general shape functions or fairness constraints; see also \cite{liebl2023fast}.

\subsection{Comparison of distributions}\label{sec:ComparisonOfDistributions}

We now opt for a more general approach and consider a simultaneous comparison of one dimensional distributions to test hypothesis $H_0^F$. To this end, let us define the available data estimator
\begin{align*}
&\hatFA(t,z) = \frac{1}{n} \frac{1}{\widehat{p}_{\widehat{A}}(t)}\sum_{i = 1}^n \ind\{X_i(t) \leq z, t\in O_i, O_i \in \widehat{\mathcal{A}}\},&& z \in \mathbb{R},t \in [0,1], 
\end{align*}
and define $\hatFB$ in an analogous way. For $t$ and $z$, it follows, under our assumptions, that
\begin{align*}
\hatFA(t,z) \xrightarrow{p} \FA(t,z) = \prob(X(t)\leq z\vert t \in O, A),  &&
\hatFB(t,z) \xrightarrow{p} \FB(t,z) = \prob(X(t)\leq z\vert t \in O, B),
\end{align*}
as $n \to \infty$. Under MCAR, $\FA(t, z) = \prob(X(t) \leq z)$ and $\FB(t, z) = \prob(X(t) \leq z)$. We therefore base a test on the comparison of $\hatFA$ and $\hatFB$. To this end, consider the product measure $ \lambda \otimes \nu$ where $\lambda$ is the Lebesgue measure on $[0,1]$ and $\nu$  some specific probability measure supported on $\Xi \subset \mathbb{R}$. We consider the following Cramér--von Mises type test statistic
\begin{equation*}
T_{F} = n \int_{[0,1]\times \Xi} \left(\hatFA(t,z) - \hatFB(t,z)\right)^2\, \dF(t,z).
\end{equation*}
Related tests have been explored by \cite{hall2007twosample}, \cite{bugni2009goodness}, and \cite{bugni2012specification}. 
The performance of the resulting test depends on the probability measure $\nu$. It can be chosen by the practitioner or in a data-adaptive way. A particular choice of $\nu$ is discussed in Section~\ref{sec:NumericalIllustration}. To study stochastic properties of~$T_{F}$, we consider the separable Hilbert space $H = L^2([0,1]\times \Xi, \lambda \otimes \nu)$. Then, $\hatFA$ and $\hatFB$ are viewed as $H$-valued random variables~$\mathcal{X}$ which satisfy $\expv[\norm{\mathcal{X}}_H^2]<\infty$. Here, $\norm{\cdot}_H$ denotes the induced norm on $H$, that is to say, $\norm{f}^2_H = \int_{[0,1]\times \Xi} f(t,z)^2 \,\dF(t,z)$.

\newpage

\begin{theorem}\label{thm:CLT_mu}
Let Assumptions~\ref{ass:Missing} and~\ref{ass:hatA} be satisfied. Under MCAR,
\begin{equation*}
T_{F} \quad \xrightarrow{d} \quad \int_{[0,1]\times \Xi} Z(t, z)^2 \, \dF(t,z),
\end{equation*}
as $n \to \infty$, where $Z$ is an $H$-valued Gaussian element with mean zero and covariance function
\begin{equation}\label{eq:CovRho}
\begin{aligned}
\rho(s,t,z_1,z_2) & = \prob(A)\, \frac{\expv[\ind^c\{X(s)\leq z_1\}\ind^c\{X(t)\leq z_2\}\xi_O(s)\xi_O(t)\vert A]}{\prob(s\in O, A)\prob(t\in O, A)} \\
&\qquad + \prob(B)\, \frac{\expv[\ind^c\{X(s)\leq z_1\}\ind^c\{X(t)\leq z_2\}\xi_O(s)\xi_O(t)\vert B]}{\prob(s\in O, B)\prob(t\in O, B)},
\end{aligned}
\end{equation}
where $\ind^c\{X(t)\leq z\} = \ind\{X(t)\leq z\}- \left(\FA(t,z)\ind\{A\}+\FB(t,z)\ind\{B\}\right)$.
\end{theorem}
The approximation of the limiting distribution in Theorem~\ref{thm:CLT_mu} is discussed in the supplementary file.

\section{Consistency}\label{sec:Consistency}

To discuss consistency of our tests, suppose first that $(X,O)$ satisfies MCAR. Define a new process by $Y(t) = X(t) + \xi_{[0,1]\setminus O}(t)$ for $t \in [0,1]$. Clearly $(Y,O)$ violates MCAR but $(Y\xi_O, O) = (X\xi_O, O)$. This shows that it is generally impossible to distinguish all alternatives from the null based on the available data. In fact, any test has trivial power against all \emph{compatible} alternatives. For a related discussion, the interested reader is referred to \citep{berrett2023optimal}. In the following, we discuss practical conditions under which our tests are consistent.

\begin{theorem}\label{thm:power} It holds that
\begin{enumerate}[label=(\roman*)]
\item if $\norm{\mA-\mB}_{\infty} > 0$ and the assumptions of Theorem~\ref{thm:CLT_D} hold, $T_\mu \xrightarrow{p} \infty$ as $n \to \infty$;
\item if $\norm{\FA-\FB}_{H} > 0$ and the assumptions of Theorem~\ref{thm:CLT_mu} hold, $T_{F} \xrightarrow{p} \infty$ as $n \to \infty$.
\end{enumerate}
\end{theorem}
The theorem establishes general consistency results for our hypothesis tests. In case of a deterministic clustering, it can be shown that the tests have power against local alternatives of the form $\norm{\mu_{A,n}-\mu_{B,n}}_\infty = \frac{c}{\sqrt{n}}$  and $\norm{F_{A,n}-F_{B,n}}_{H} = \frac{c}{\sqrt{n}}$, for $c > 0$; see also the remark after the proof of Theorem~\ref{thm:power}. We now discuss general conditions which yield $\norm{\FA-\FB}_{H} > 0$.

\begin{proposition}\label{prop:ConsistencyDist}
Suppose that there exists some $(t_0, z_0) \in \textup{int}\left([0,1]\times \Xi \right)$ such that
\begin{enumerate}[label=(\roman*)]
    \item $(t,z) \mapsto F_A(t,z)-F_B(t,z)$ is continuous at $(t_0, z_0) $;
    \item $ \{O \in \mathcal{A}\} \not\perp\!\!\!\!\perp \{X(t_0) \leq z_0\}  \mid \{t_0 \in O\}$.
\end{enumerate}
Then it holds that $\norm{\FA-\FB}_{H} > 0$.
\end{proposition}

The crucial condition is (ii) which states that, conditional on $\{t_0 \in O\}$, the event $\{O\in \mathcal{A}\}$ depends on $\{X(t_0)\leq z\}$. For this class of alternatives we obtain $\norm{\FA-\FB}_{H} > 0$ and Theorem~\ref{thm:power} then implies that our distribution test is consistent. An example which satisfies (ii) is given in Example~\ref{ex:BMcensored} below.

\begin{proposition}\label{prop:consistency}
Suppose that Assumption~\ref{ass:Missing} holds, that $\mA(t)-\mB(t)$ is continuous and that $\norm{\mA-\mB}_{\infty} > 0$. If $\nu$ is supported on $\Xi$ and for each $t \in [0,1]$ it holds that $X(t) \in \textup{int}(\Xi)$ almost surely, then $\norm{\FA-\FB}_H > 0$.
\end{proposition}
The above proposition shows that the distribution test is more general in the sense that it also detects differences in means provided that the measure $\nu$ is chosen properly. To obtain further insights, consider the following example.

\begin{example}\label{ex:BMcensored}
Let $X = (B(t)\colon [0,1])$ be a Brownian motion and assume the censoring mechanism $O = \{t \in [0,1]\colon X(t) \in [a,b]\}$, for some $a < 0 < b$. Consider the partition given by $\mathcal{A} = \{[0,1]\}$ which corresponds to complete versus incomplete sets. It is easy to see for $a = -b$, the censoring does not affect the mean and hence we have $\norm{\mA-\mB}_{\infty} = 0$. However, using Proposition~\ref{prop:ConsistencyDist}, it can be shown that $\norm{\FA-\FB}_H > 0$. 
\end{example}
Note that in Example~\ref{ex:BMcensored}, $\norm{\FA-\FB}_H > 0$ actually holds for arbitrary $a < 0 < b$. Together with Theorem~\ref{thm:CLT_mu}, it shows that a simple clustering of complete versus incomplete sets is sufficient for~$T_F$ to detect the violation of MCAR. Moving towards general $a \neq -b$, it is also possible to get consistency of $T_{\mu}$. An explicit computation of $\norm{\mA-\mB}_\infty $ seems out of reach for  the general case; different types of arguments are needed for different pairs $a$ and~$b$. We focus here on the case, which presumably is the most difficult to detect for $T_{\mu}$, namely where $a=-b+\varepsilon$, with small $\varepsilon$ (almost symmetric case) and where $b$ is large ($\mA$ and $\mB$ are very close to zero). 

\begin{proposition}\label{prop:BrownianMean}
Take the same scenario as in Example~\ref{ex:BMcensored}. Then there exist $b_0>0$ and $\varepsilon_0>0$, such that if $a = -b + \varepsilon$ for $b>b_0$ and $0<\varepsilon<\varepsilon_0$, then $\norm{\mA-\mB}_{\infty}>0$.
\end{proposition}
\noindent The results will be further illustrated in the numerical simulations of the next section. Figure~\ref{fig:examplecensored} shows censored paths along with the difference $\widehat{\mu}_A - \widehat{\mu}_B$ in empirical means for the considered case of a Brownian motion and $(a,b) = (-1,2)$.

\section{Numerical illustration}\label{sec:NumericalIllustration}

The methodology is implemented in R \citep{R}. We assume that curves are measured on a regular grid and approximate integrals by discrete summations. In our real-data applications, the partitioning of the observation sets is carried out as described in Section~\ref{sec:Clustering}, and the subdomain $I$ in Remark~\ref{rem:I} is defined as the intersection of the cluster centers; see Proposition~\ref{prop:ChoiceOfI}. For the simulations, we also consider other clustering approaches and choose $I \subset [0,1]$  such that, in each group, at least 25\% of the curves are observed on $I$. The measure $\nu$ of our test statistic $T_{F}$ is set to be a Gaussian probability measure, $\nu \sim N(\theta, \tau^2)$, with $\theta = \int_0^1 \expv[X(t)]\,\dt$ and $\tau^2 = \sup_{t \in [0,1]} \var(X(t))$. The parameters $\theta$ and $\tau^2$ are estimated under MCAR from the available data.

\subsection{Simulation study}\label{sec:SimulationStudy}

We repeatedly simulate $n$ independent copies $X_1, \dots, X_n$ of a standard Brownian motion $X$ evaluated on a grid of $100$ regularly spaced points in $[0,1]$. For checking numerical properties of the test procedures, we then estimate error probabilities over 1,000 such simulation runs. Additional results for non-Gaussian data are presented in the supplementary file.

\textit{Case 1: Missing completely at random.} First, we consider a situation in which the MCAR assumption is satisfied and simulate $O$ independently from $X$. To this end, let $U_1$ and $U_2$ be independent uniformly distributed random variables on $[0,1]$ and define $L = \min\{U_1,U_2\}$ and $U = \max\{U_1, U_2\}$. We then set $O = 
[0,1]$ with probability 1/2 and $O = [L,U]$ otherwise. Due to the missingness mechanism, about half of the curves are observed completely and the mean length of the observation sets equals $\expv[\lambda(O)] = 2/3$. Table~\ref{tab:MCAR} presents the estimated probability of a type I error for a nominal level of $\alpha = 0.05$. For comparative purposes, we also consider a simple partition of complete versus incomplete observation sets and an alternative variant which approximates the test distribution using a bootstrap. Details on our bootstrap procedure are provided in the supplementary file.

{
\begin{table}[h]
\centering
\begin{tabular}{|c|cc|cc|cc|cc|}
\cline{2-9}
\multicolumn{1}{c|}{} & \multicolumn{4}{c|}{systematic clustering} & \multicolumn{4}{c|}{complete vs.~incomplete} \\
\cline{2-9}
\multicolumn{1}{c|}{} & \multicolumn{2}{c|}{asymp} & \multicolumn{2}{c|}{boot} & \multicolumn{2}{c|}{asymp} & \multicolumn{2}{c|}{boot} \\
\hline
$n$ & $T_\mu$ & $T_{F}$ & $T_\mu$ & $T_{F}$ & $T_\mu$ & $T_{F}$ & $T_\mu$ & $T_{F}$ \\
\hline
100 & 0.069 & 0.085 & 0.043 & 0.069 & 0.065 & 0.065 & 0.049 & 0.057 \\
250 & 0.068 & 0.056 & 0.051 & 0.049 & 0.059 & 0.050 & 0.048 & 0.049 \\
500 & 0.057 & 0.051 & 0.045 & 0.050 & 0.053 & 0.045 & 0.047 & 0.044 \\
\hline
\end{tabular}
\caption{Estimated type I error probabilities in Case~1 for a nominal level of~$\alpha = 0.05$. }\label{tab:MCAR}
\end{table}}

\textit{Case 2: Censoring.} Motivated by the theoretical results in Section~\ref{sec:Consistency}, we now examine the performance of the tests for data where MCAR is violated.  We consider the following scenarios:
\begin{flushleft}
{\renewcommand{\arraystretch}{1}
\begin{tabular}{ll}
    (MNAR) & $O = \{t \in [0,1]\colon X(t) \in [a,b]\}$, \\
    (MAR) &$O = [0, \tau_{ab}]$, where $\tau_{ab} = \inf\{t \geq 0\colon X(t) \in \{a,b\}\}$,
\end{tabular}
}
\end{flushleft}
\noindent for $a = -1$, $b \in \{1.00, 1.10, \dots, 2.00\}$. Note that the second scenario satisfies the definition of MAR since $\tau_{ab}$ is a stopping time with respect to the canonical filtration of $X$. The setting $(a,b) = (-1,1)$ corresponds to symmetric censoring in which case the conditional mean functions corresponding to complete and incomplete sets are both equal to zero. See also Example~\ref{ex:BMcensored} and the subsequent discussion. In this situation, $T_\mu$ should reveal lower power as it is based on a comparison of means. Higher values of $b$ relate to a stronger deviation from the symmetric case which favors the mean comparison. In contrast, since symmetric censoring affects conditional distributions, $T_{F}$ is expected to have power even in the case $(a,b) = (-1,1)$.

Estimated rejection probabilities are presented in the first row of Figure~\ref{fig:censoredrejections}. In scenario (MAR), we also compare our methodology to the sequential multiple hypothesis test $T_{\text{LR}}$ proposed by \cite{liebl2019partially}, which has been formulated for observation intervals $O = [0, D]$. The plot confirms our theory whereby the distributional test~$T_{F}$ is more powerful than the mean comparison~$T_\mu$. The latter requires some asymmetry to detect the violations and approaches the significance level $\alpha = 0.05$ in the case $(a,b) = (-1,1)$. In Figure~\ref{fig:examplecensored} (left), we present a particular sample of the simulation study for the case of strong asymmetry with $(a,b) = (-1, 2)$. Figure~\ref{fig:examplecensored} (right) reveals the violation of the MCAR hypothesis: for~$t \in [0.40, 1]$, the conditional mean in the complete sample is significantly greater than the one computed on the incomplete sample only.

To examine the sensitivity with respect to the selection of the groups, we consider different deterministic choices that partition the observation sets according to their size. Specifically, we define $\mathcal{A} = \{O \in \mathcal{K}\colon \lambda(O) < \delta\}$ for $\delta \in \{0.50, 0.60, \dots, 1.00\}$. Rejection probabilities of the distribution test for $(a,b) = (-1,1)$ are shown in the second row of Figure~\ref{fig:censoredrejections}. The reference line denotes the mean rejection probability of the systematic clustering approach.
In the particular setting of scenario~(MNAR), partitioning the data into complete versus incomplete ($\delta = 1$) yields the highest power and clearly outperforms the systematic procedure for $n = 100$. This is in contrast to the setting considered in scenario~(MAR), where the systematic procedure reveals a rejection probability that is close to the optimal grouping which corresponds to $\delta = 0.6$. In both settings, the power increases with the sample size across all partitions.

\begin{figure}
\centering
\begin{subfigure}[b]{0.49\textwidth}
\caption*{Scenario (MNAR)}
\includegraphics[width=\textwidth]{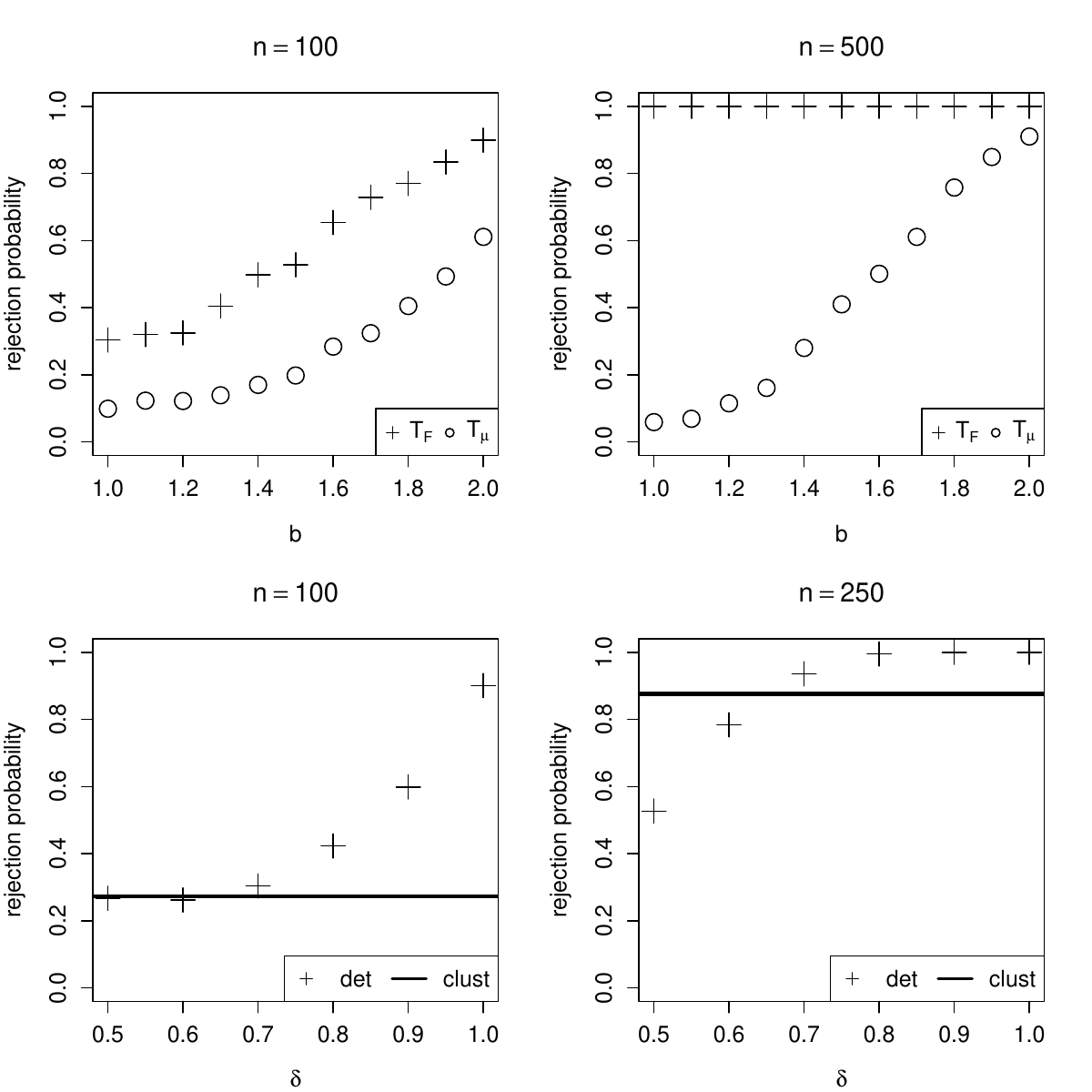}
\end{subfigure}
\begin{subfigure}[b]{0.49\textwidth}
\caption*{Scenario (MAR)}
\includegraphics[width=\textwidth]{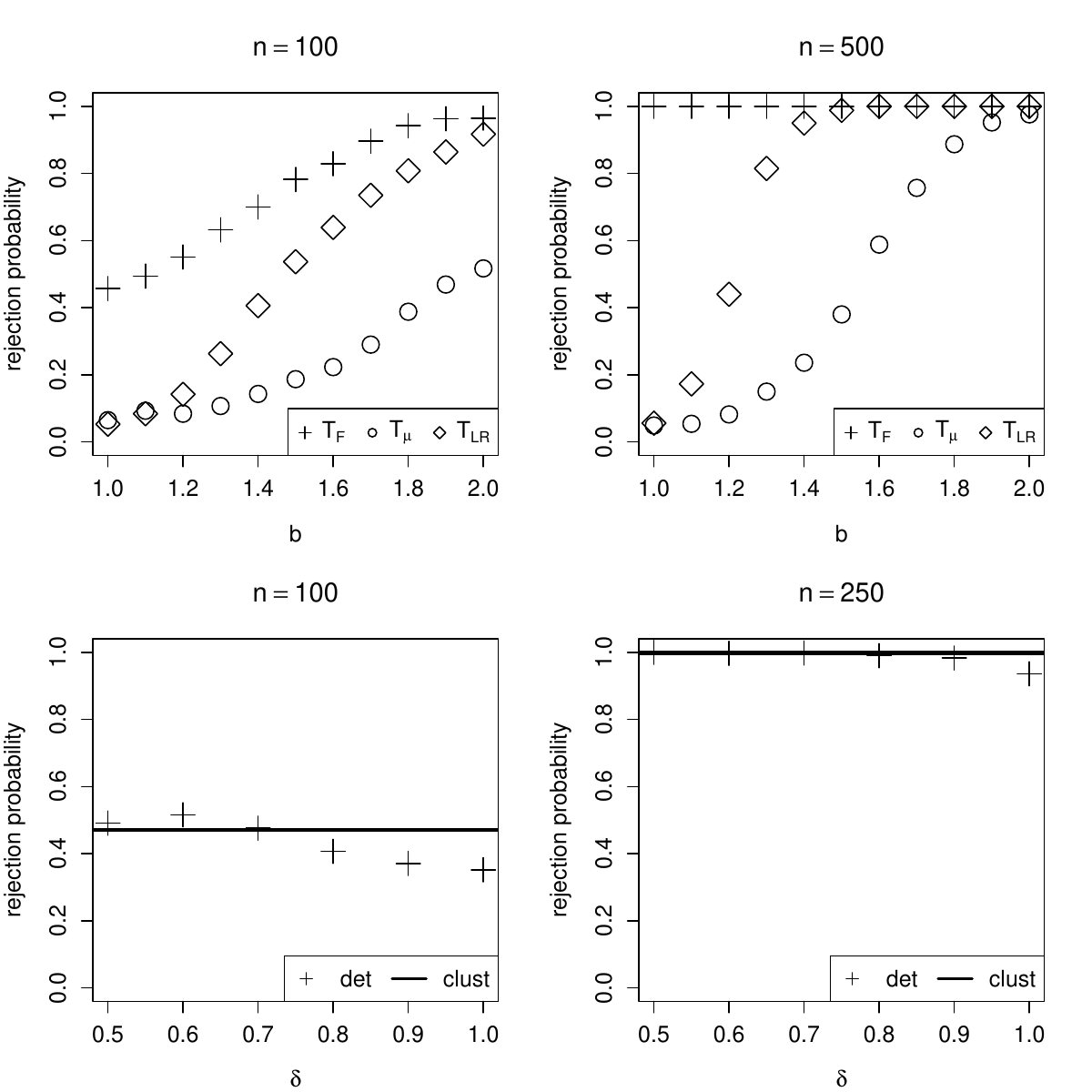}
\end{subfigure}
\caption{Rejection probabilities in Case~2 using tests based on asymptotic distributions. Top row:~comparison of $T_F$ and $T_\mu$.  Bottom row:~comparison of systematic and deterministic clustering with respect to the length $\delta$ of the observable domain.}
\label{fig:censoredrejections}
\end{figure}

\begin{figure}
\centering
\includegraphics[width=0.8\linewidth]{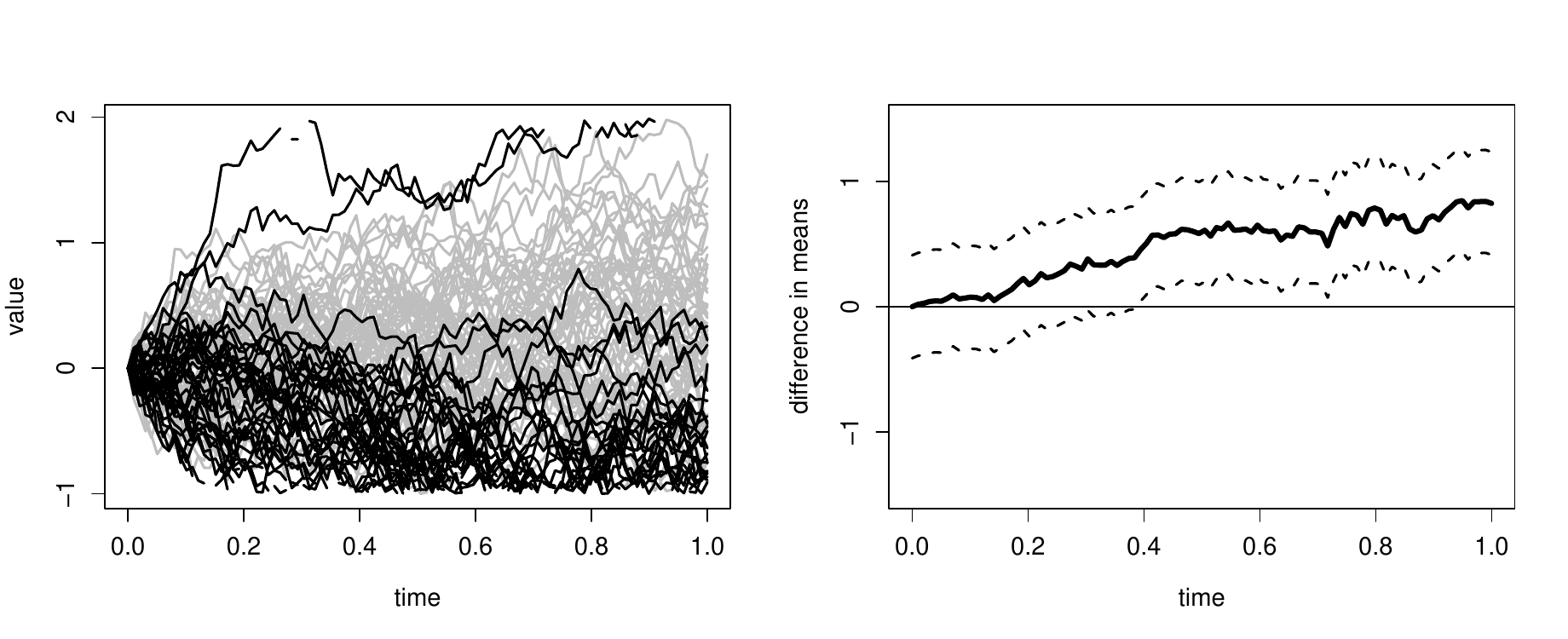}
\caption{Sample of $n = 100$ paths of a Brownian motion with (left) and the estimated difference $\widehat{\mu}_A - \widehat{\mu}_B$ in means~(right). Curves are grouped into complete $(\mathcal{A})$ versus incomplete $(\mathcal{B})$ and observed if $X(t) \in [-1, 2]$. The dashed line relates to a simultaneous confidence band with coverage probability~95\% (see Algorithm~\ref{alg:SCB} in the supplementary file).}
\label{fig:examplecensored}
\end{figure}

\subsection{Heart rate data}\label{sec:HeartRateData}

Next, we apply our methodology to incompletely observed heart rate profiles from the ``Swiss Kidney Project on Genes in Hypertension" \citep{alwan2014epidemiology} considered in the works \cite{kraus2015components} and \cite{kraus2019inferential}. The data consist of $n = 878$ aligned curves recorded on an equispaced grid of $361$ points in the time interval~$[8~\text{pm}, 2~\text{am}]$. Here $X_i(t)$ refers to the heart rate (in beats per minute) of a subject $i$ at time~$t$. In 527 cases (around~60\%), observations are complete. Following \cite{kraus2015components}, missing data is either caused by participants that remove measuring devices due to discomfort, or failure of the devices. According to experts, a MCAR assumption would be plausible. Our methodology does not detect a significant violation of MCAR (see Table~\ref{tab:ResultRealData}). The sample of heart rate profiles is shown in the first row of Figure~\ref{fig:examples}. In addition, Figure~\ref{fig:examples} visualizes the observation patterns in form of a ``barcode plot" as well as the mean difference $\hatmA - \hatmB$ together with a simultaneous confidence band obtained by Algorithm~\ref{alg:SCB} discussed in the supplementary file. The difference curve does not reveal any departures from the zero line.

\subsection{Electricity data}\label{sec:electricity}

We reconsider an example from the German electricity market data originally presented in \cite{liebl2019partially}. The data set consists of $n = 583$ weekly, partially observed, and pre-smoothed price curves $X_i$, which represent the logarithmized price (EURO/MW) as a function of electricity demand (MW). We focus on the domain $[1700, 2500]\text{ MW}$, where missing curves occur, and consider $161$ equispaced points within this range. Each price curve is initially observed over the beginning of the domain but may become missing over the later part. Once a curve $X_i$ becomes missing, it remains so. Specifically, for each curve, there exists some $D_i\in(1700, 2500] \text{ MW}$ such that the observation set satisfies $O_i = [1700, D_i]\text{ MW}$; see the left plot in the second row of Figure~\ref{fig:examples}. This yields a monotone missingness pattern as considered in Example~\ref{ex:[0,D]} and our systematic approach thus clusters the curves according to the extent of their observations.

The structure of the electricity market suggests that price curves observed over larger subdomains (which correspond to a larger $D_i$) tend to exhibit overall higher price levels, and vice versa. Consequently, the missing completely at random (MCAR) assumption may be violated in this context. Our tests confirm this intuition and strongly reject the MCAR hypothesis (see Table~\ref{tab:ResultRealData}). Additionally, Figure~\ref{fig:examples} illustrates the difference in mean curves $\hatmA-\hatmB=(\hatmA(t)-\hatmB(t)\colon t\in I)$ along with the corresponding $95 \%$ simultaneous confidence band (Algorithm \ref{alg:SCB} of the supplementary file), which significantly deviates from the zero line across the subdomain $I$.

\subsection{Temperature data}\label{sec:TemperatureData}

Finally, we consider half-hourly recorded temperature curves from Graz (Austria). The data are provided by \cite{data2023} and have initially been considered in \cite{ofner2024covariateinformed}. Out of $n = 76$ daily curves, $66$ are observed on the full domain. The remaining~10 days contain missing values in the second half of the day; see the third row of Figure~\ref{fig:examples} for a visualization of the data set. According to the monitoring executive, the cause of missingness is unclear and can potentially be due to the sensors' sensitivity to high temperature values or caused by a power problem in the transmission unit. The former relates to censoring and would resemble a clear violation of MCAR. Testing MCAR could thus help the monitoring agency to locate the error in the measuring device. The test distributions are approximated by a bootstrap procedure (see~Section~\ref{sec:Bootstrap} of the supplementary file). Because of the small $p$-values (see Table~\ref{tab:ResultRealData}), an issue with the temperature sensor is plausible and should be further investigated.

\begin{figure}[!htp]
\thisfloatpagestyle{empty}
\begin{center}
\begin{tabular}{|c|cccrrrr|}
\hline
data        & $n$ & $\vert \widehat{\mathcal{A}}\vert$ & distribution & $p_{\mu}$ & $p_F$ & $T_\mu$ &  $T_F$ \\
\hline
heart rate  & 878 & 809 & asymptotic &  0.765   & 0.911   & 42.49 &  0.97 \\
electricity & 583 & 302 & asymptotic &  $<0.001$ & $<0.001$ & 47.76  & 82.40\\
temperature & 76  & 68  & bootstrap  &  $0.036$ & $0.054$  & 32.59  & 2.10 \\
\hline
\end{tabular}
\end{center}
\captionof{table}{Summary of hypothesis tests.}\label{tab:ResultRealData}

\vspace{0.5cm}

\centering
    \begin{subfigure}[b]{0.15\textwidth}
        \centering
        \caption*{$\mathcal{A}$}
        \vspace{-0.6cm}
        \includegraphics[width=\textwidth]{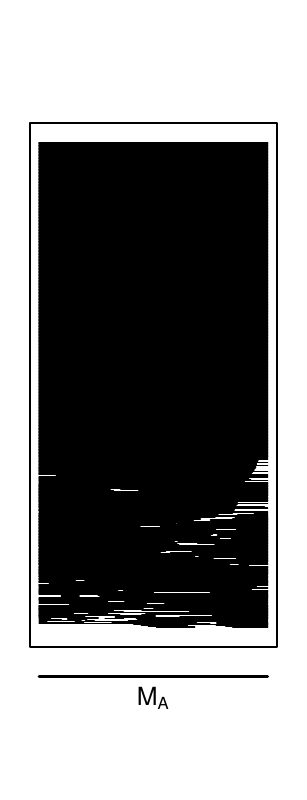}
    \end{subfigure}
    \begin{subfigure}[b]{0.15\textwidth}
        \centering
        \caption*{$\mathcal{B}$}
        \vspace{-0.6cm}
        \includegraphics[width=\textwidth]{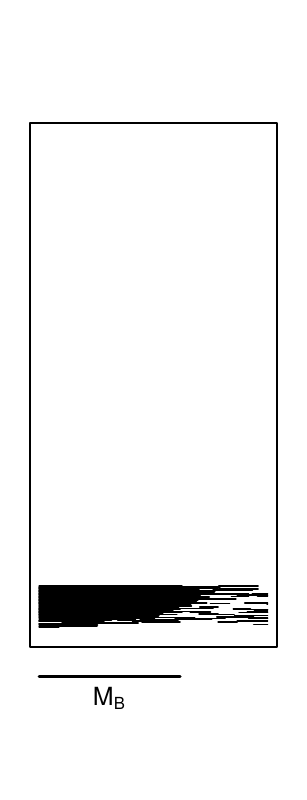}
    \end{subfigure}
    \begin{subfigure}[b]{0.3\textwidth}
    \centering
    \caption*{sample}
        \vspace{-0.6cm}
        \includegraphics[width=\textwidth]{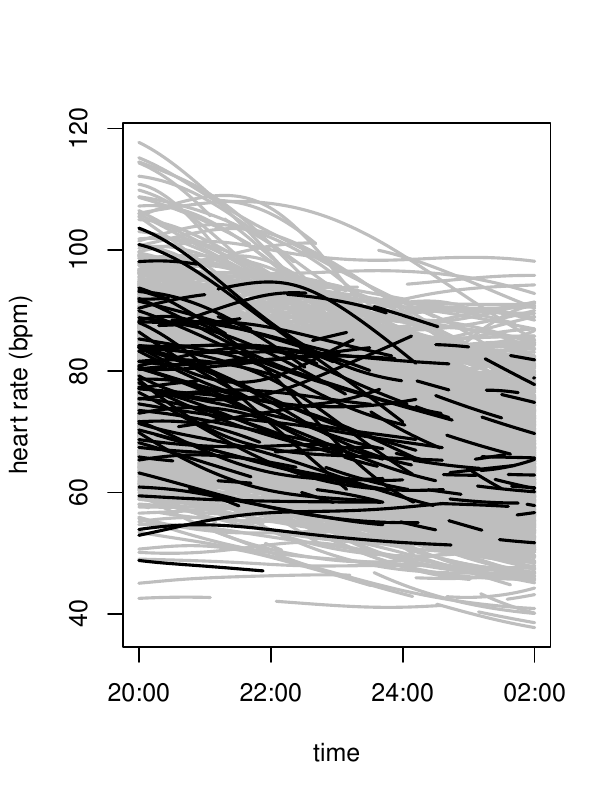}
    \end{subfigure}
    \begin{subfigure}[b]{0.3\textwidth}
    \centering
    \caption*{mean difference}
    \vspace{-0.6cm}
    \includegraphics[width=\textwidth]{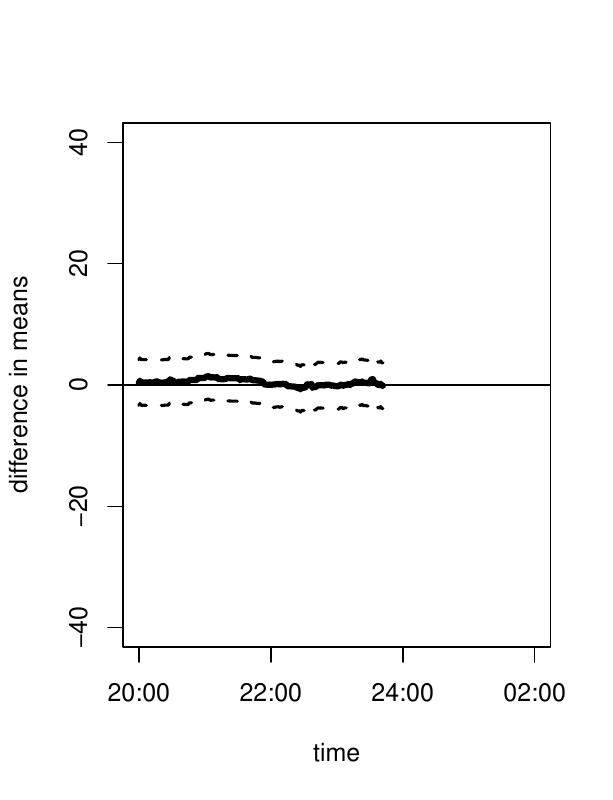}
    \end{subfigure}
    
    \begin{subfigure}[b]{0.15\textwidth}
        \centering
        \includegraphics[width=\textwidth]{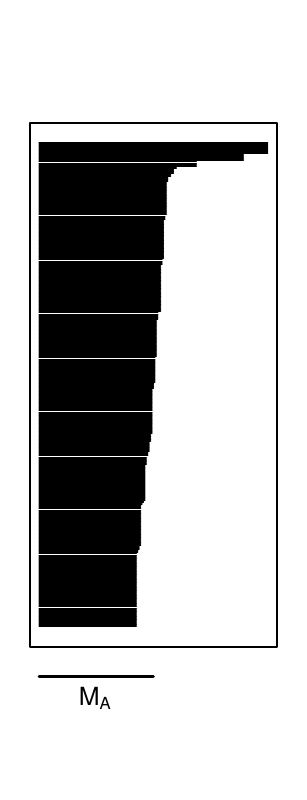}
    \end{subfigure}
    \begin{subfigure}[b]{0.15\textwidth}
        \centering
        \includegraphics[width=\textwidth]{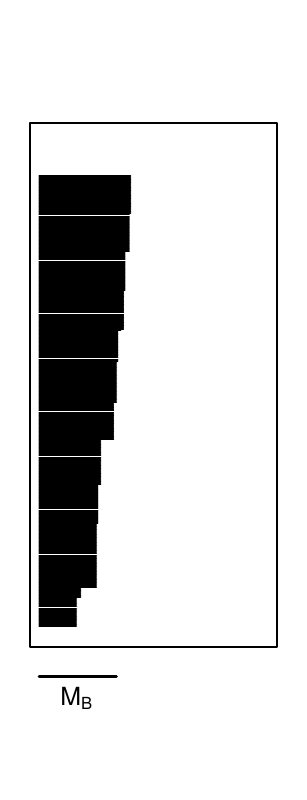}
    \end{subfigure}
    \begin{subfigure}[b]{0.3\textwidth}
        \centering
        \includegraphics[width=\textwidth]{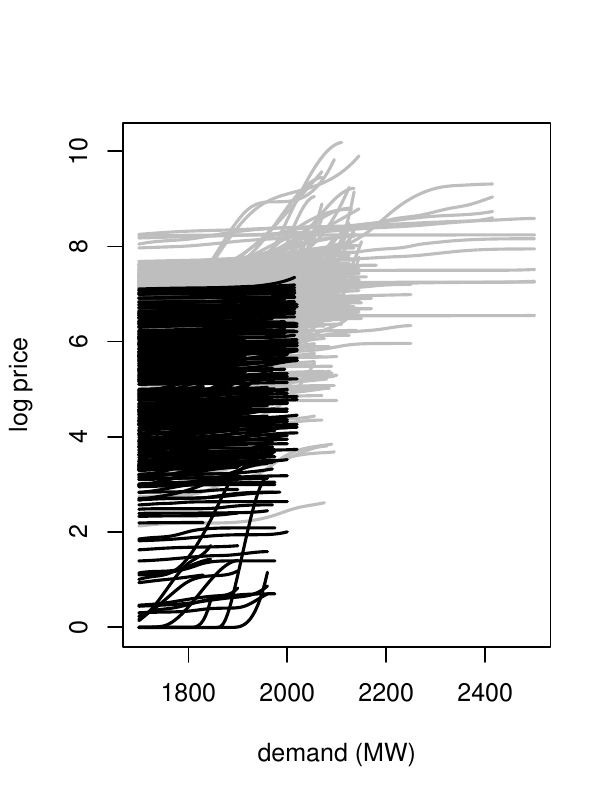}
    \end{subfigure}
    \begin{subfigure}[b]{0.3\textwidth}
        \centering
        \includegraphics[width=\textwidth]{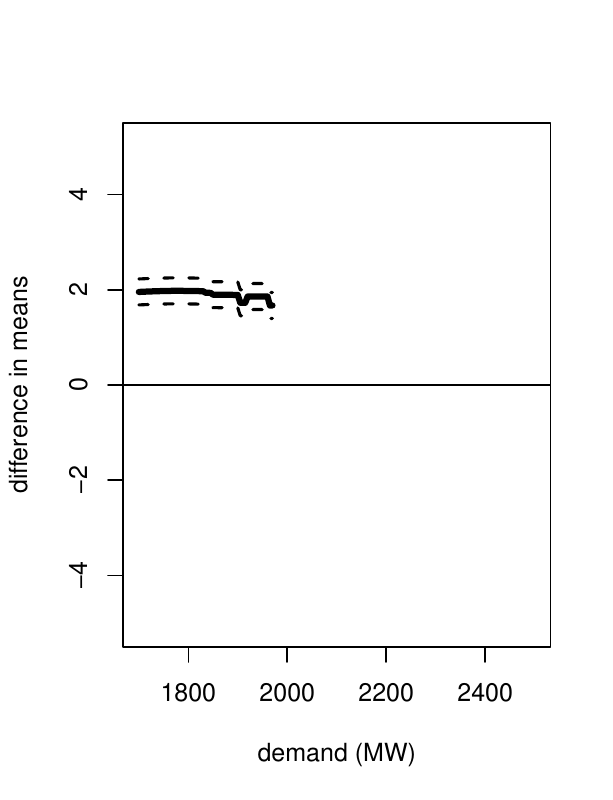}
    \end{subfigure}
    
    \begin{subfigure}[b]{0.15\textwidth}
        \centering
        \includegraphics[width=\textwidth]{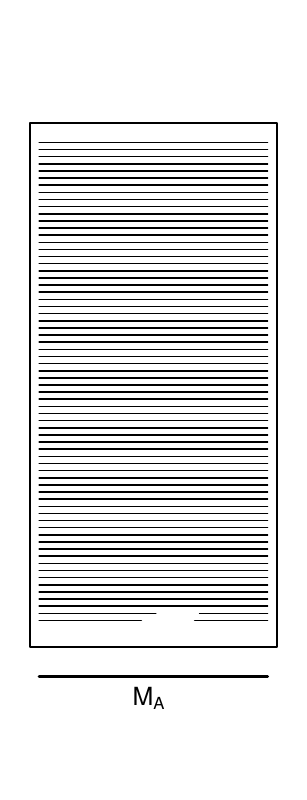}
    \end{subfigure}
    \begin{subfigure}[b]{0.15\textwidth}
        \centering
        \includegraphics[width=\textwidth]{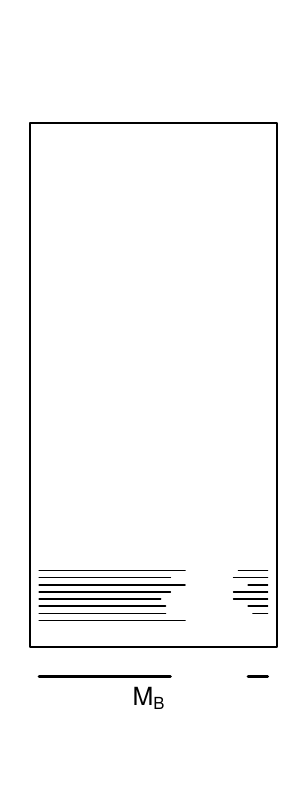}
    \end{subfigure}
    \begin{subfigure}[b]{0.3\textwidth}
        \centering
        \includegraphics[width=\textwidth]{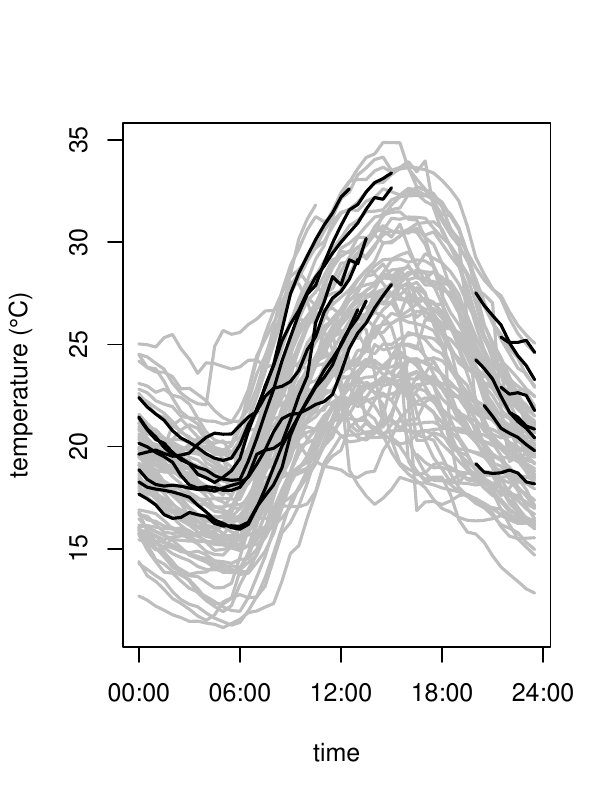}
    \end{subfigure}
    \begin{subfigure}[b]{0.3\textwidth}
        \centering
        \includegraphics[width=\textwidth]{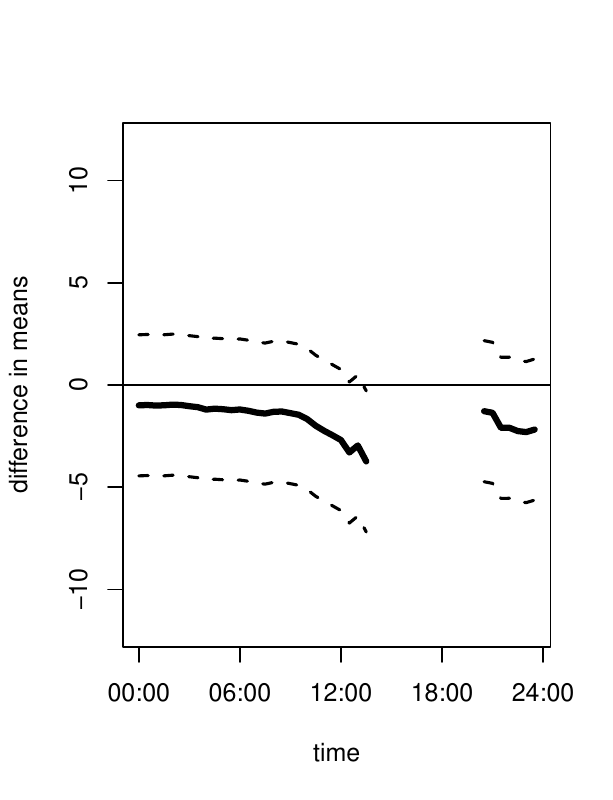}
    \end{subfigure}
\captionof{figure}{Sample of heart rate profiles (first row), electricity data (second row), and temperature data (third row). The first column visualizes the observation patterns, the second column the curves ($\widehat{\mathcal{A}}\colon$~gray, $\widehat{\mathcal{B}}\colon$~black), and the third column the estimated difference in means $\hatmA-\hatmB$  over the domain $I$ along with 95\% simultaneous confidence bands.}
\label{fig:examples}

\end{figure}

\section{Discussion}

We have proposed a novel framework for testing MCAR in case of functional data which is an intricate problem as the number of different observation patterns is generally unbounded. Our method relies on a partition of the observation sets into two groups. These groups can either be chosen systematically or ad hoc by the practitioner. Simulation results indicate that the methodology is reasonably insensitive to the selection of groups provided that the sample size is large enough. For smaller sample sizes, employing an expert-guided deterministic selection approach could improve the performance. Concerning the comparison of groups, we considered approaches based on means and distributions. While the mean comparison allows for informative plots, the distribution test is generally recommended as both theoretical and numerical results suggest that it is consistent against a larger class of alternatives.
Another promising extension of our methodology is to investigate supervised classification methods, rather than unsupervised approaches such as $k$-medians. When suitable training data are available, supervised classification has the potential to increase the power of our testing procedure.

\section*{Funding}

The work of D.~Liebl is funded by the Deutsche Forschungsgemeinschaft (DFG, German Research Foundation) under Germany’s Excellence Strategy – EXC-2047/1 – 390685813. Siegfried Hörmann was supported by the Austrian Science Fund (FWF), project 10.55776/P35520.

\section*{Acknowledgments}

We thank anonymous referees for several valuable remarks, which lead to a significant improvement of an earlier version of this manuscript. The computational results presented have been achieved using the Vienna Scientific Cluster (VSC).

\putbib[ref]
\end{bibunit}

\newcommand{\beginsupplement}{
    \setcounter{table}{0}
    \renewcommand{\thetable}{S\arabic{table}}
    \setcounter{figure}{0}
    \renewcommand{\thefigure}{S\arabic{figure}}
    \setcounter{equation}{0}
    \setcounter{page}{1}
    \setcounter{section}{0}
    \renewcommand{\thesection}{S\arabic{section}}
}

\clearpage
\beginsupplement

\begin{bibunit}
\begin{center}
    \textit{Supplementary Material for:} \\
    \vspace{0.4cm}
    {\Large\textbf{Testing the Missing Completely at Random Assumption for Functional Data}} \\
    \vspace{0.6cm}
    By Maximilian Ofner, Siegfried H\"ormann, David Kraus \& Dominik Liebl \\
    \vspace{0.4cm}
\end{center}

\section{Proofs}

\subsection{Proofs of results in Section~\ref{sec:Clustering}}

\begin{proof}[of Proposition~\ref{prop:ChoiceOfI}]
The set $M_A$ minimizes $M_1 \mapsto \expv[d(M_1,O) \vert A]$. Indeed, if there existed some $M_A^*$ such that 
\begin{equation*}
    \expv[d(M_A^*,O) \vert A] < \expv[d(M_A,O) \vert A],
\end{equation*}
then
\begin{align*}
\expv[\min\{d(M_A^*,O) , d(M_B, O)\}] &< \expv[d(M_A,O) \ind\{A\}] + \expv[d(M_B, O)\ind\{B\}] \\
&= \expv[\min\{d(M_A,O) ,d(M_B, O)\}],
\end{align*}
which results in a contradiction. Since $M_A$ minimizes
\begin{align*}
    \expv[d(M_1, O)\vert A] &= \int_0^1 \expv [\abs{\xi_{M_1}(t)-\xi_O(t)}^2 \vert A] \,  \dt \\
    &= \int_0^1 \xi_{M_1}(t)(1-2\prob(t \in O\vert A))\,\dt + \int_0^1 \prob(t \in O\vert A) \, \dt,
\end{align*}
it must hold that $M_A = \{t\colon \prob(t\in O\vert A) \geq 1/2\}$. Similarly, $M_B = \{t\colon \prob(t\in O\vert B) \geq 1/2\}$ which finally establishes
\begin{equation*}
    \inf_{t \in I} \{\prob(t \in O\vert A), \prob(t \in O\vert B)\} \geq 1/2,
\end{equation*}
for $I = M_A \cap M_B$. 
\end{proof}

\begin{proof}[of Proposition~\ref{prop:ConsistencyCenters}]
The proof is based on Theorem~5.7 of \cite{Vaart1998Asymptotic}. To this end, let $\mathcal{M}_0^{(2)} = \{\{M_1, M_2\}\colon M_1, M_2 \in \mathcal{M}_0\}$ denote the set of unordered pairs of $\mathcal{M}_0$ equipped with the metric
\begin{equation*}
    d^{(2)}(\{M_1, M_2\}, \{M_1', M_2'\}) = \min \{d(M_1, M_1') + d(M_2, M_2'), d(M_1, M_2') + d(M_2, M_1')\}.
\end{equation*}
Since $\mathcal{M}_0$ is $d$-compact, $\mathcal{M}_0^{(2)}$ is $d^{(2)}$-compact. Moreover, the mapping
\begin{equation*}
    \{M_1, M_2\} \mapsto \min\{d(O, M_1), d(O, M_2)\},
\end{equation*}
is $d^{(2)}$-continuous for every $O$ and bounded by one. The conditions of Theorem~5.7 in \cite{Vaart1998Asymptotic} can then be verified using  Example 19.8 in \cite{Vaart1998Asymptotic} and the uniqueness of the minimizer $\{M_A, M_B\}$. The theorem then yields
\begin{equation*}
    d^{(2)}(\{\widehat{M}_A, \widehat{M}_B\}, \{M_A, M_B\}) = o_p(1).
\end{equation*}
Since $d(M_A, M_B)$ is assumed to be positive and $d(\widehat{M}_A, M_A) \leq d(\widehat{M}_B, M_A)$, it follows that $d(\widehat{M}_{A}, M_{A}) +d(\widehat{M}_{B}, M_{B}) = o_p(1)$.
\end{proof}

\begin{proof}[of Lemma~\ref{lem:AmathcalA}]
First, we note that $\prob(O_i \in \mathcal{A} \sym \widehat{\mathcal{A}}) = \prob(O_i \in \mathcal{B} \sym \widehat{\mathcal{B}})$ by properties of the symmetric difference. It then suffices to show that $\prob(O_i \in \mathcal{A} \setminus \widehat{\mathcal{A}}) = o(1)$; the remaining term $\prob(O_i \in \widehat{\mathcal{A}} \setminus \mathcal{A})$ can be treated similarly. For any $\delta>0$, $\prob(O_i \in \mathcal{A} \setminus \widehat{\mathcal{A}})$ is bounded by
\begin{equation*}
 \prob\left(\abs{d(O,M_A) - d(O, M_B)} < 2 \delta\right)+
 \prob\left(d( M_A, \widehat{M}_A)> \delta\right) + \prob\left(d( M_B, \widehat{M}_B)>\delta  \right).
\end{equation*}
The proof then follows from a simple $\varepsilon$-$\delta$ argument.
\end{proof}

\subsection{Proof of Theorem~\ref{thm:CLT_D}}

The proof of Theorem~\ref{thm:CLT_D} is similar to the one of Theorem~\ref{thm:CLT_mu} below but care needs to be taken as the central limit theorem in $D[0,1]$ requires additional assumptions. Our strategy is based on the following result of Hahn \citep[Theorem 2]{hahn1978central}.

\begin{theorem}\label{thm:Hahn}
Let $Y$ be a $D[0,1]$-valued random variable such that $\expv[Y(t)] = 0$ and $\expv[Y^2(t)] < \infty$ for all $t \in [0,1]$. Assume there exist $\gamma_1 > 1/2$ and $\gamma_2 > 1$ such that for all $0 \leq s \leq t \leq u \leq 1$,
\begin{enumerate}[label=(\roman*)]
\item $\expv[(Y(s)-Y(t))^2] \leq C \abs{s - t}^{\gamma_1}$,
\item $\expv[(Y(s)-Y(t))^2(Y(t)-Y(u))^2] \leq C \abs{u - s}^{\gamma_2}$.
\end{enumerate}
Then $Y$ satisfies the central limit theorem in $D[0,1]$ and the limiting Gaussian process has  continuous sample paths.
\end{theorem}

\noindent The following lemma is immediate.
\begin{lemma}\label{lem:XZ}
Define $Y(t) = (X(t)-\expv{[X(t)]})Z(t)$, where $X$ and $Z$ are independent stochastic processes satisfying Assumption~\ref{ass:Y} with parameters $\alpha_X$ and $\alpha_Z$, respectively. Then there exists some constant $C_0 > 0$ such that for all $0 \leq s \leq t \leq u \leq 1$,
\begin{enumerate}[label=(\roman*)]
\item $\expv[(Y(s)-Y(t))^2] \leq C_0 \abs{s-t}^{\beta/2}$,
\item $\expv[(Y(s)-Y(t))^2(Y(t)-Y(u))^2] \leq C_0 \abs{s-u}^{\beta}$,
\end{enumerate}
where $\beta = \min\{\alpha_X, \alpha_Z\}$.
\end{lemma}

\begin{proof}[of Theorem~\ref{thm:CLT_D}]

By Lemma~\ref{lem:hatA}(ii), $\sqrt{n}\,\norm{\widehat{\mu}_{A}-\widehat{\mu}_{\widehat{A}}}_\infty + \sqrt{n}\,\norm{\widehat{\mu}_{B}-\widehat{\mu}_{\widehat{B}}}_\infty = o_p(1)$. It thus suffices to derive the limiting distribution of $\sqrt{n} \, (\widehat{\mu}_A-\widehat{\mu}_B)$. Observe $\expv[(X(t)-\mu_A(t))\xi_O(t)\ind\{O \in \mathcal{A}\}] = 0$. Therefore, the law of large numbers for $D[0,1]$-valued random variables (Theorem~1 in \cite{rao1963law}) implies
\begin{equation}\label{eq:LLN_X}
\sup_{t \in [0,1]} \Big\vert \frac{1}{n} \sum_{i=1}^n (X_i(t)-\mu_A(t))\xi_{O_i}(t)\ind\{O_i \in \mathcal{A}\} \Big\vert = o_p(1).
\end{equation}
Furthermore, since $\xi_O$ satisfies Assumption~\ref{ass:Y} and thus the conditions of Theorem~\ref{thm:Hahn},
\begin{equation}\label{eq:LLN_O}
\sup_{t \in [0,1]} \abs{\widehat{p}_A(t) - p_A(t)} = \sup_{t \in [0,1]} \Big\vert \frac{1}{n} \sum_{i=1}^n \xi_{O_i}(t)\ind\{O_i \in A\} - \prob(t\in O, A) \Big\vert = O_p(1/\sqrt{n}).
\end{equation}
From \eqref{eq:LLN_X} and \eqref{eq:LLN_O}, it follows that
\begin{equation*}
\sqrt{n} (\widehat{\mu}_A(t) -\mu_A(t)) = \frac{1}{\sqrt{n}} \sum_{i = 1}^n \frac{(X_i(t)-\mu_A(t))\xi_{O_i}(t)\ind\{O_i \in \mathcal{A}\}}{p_A(t)} + o_p(1).
\end{equation*}
Analogous results hold if $\mathcal{A}$ is replaced with $\mathcal{B}$. Since $\mu_A = \mu_B = \mu$ under MCAR,
\begin{equation*}
\sqrt{n} (\widehat{\mu}_A(t) - \widehat{\mu}_B(t)) = \frac{1}{\sqrt{n}} \sum_{i = 1}^n (X_i(t)-\mu(t))\xi_{O_i}(t)\left(\frac{\ind\{O_i \in \mathcal{A}\}}{p_A(t)} - \frac{\ind\{O_i \in \mathcal{B}\}}{p_B(t)}\right) + o_p(1),
\end{equation*}
where the remainder term is uniform in $t \in [0,1]$. Now define $Y(t) = (X(t)-\mu(t)) Z(t)$ where
\begin{equation}
Z(t) = \xi_O(t)\left(\frac{\ind\{O \in \mathcal{A}\}}{p_A(t)} - \frac{\ind\{O\in \mathcal{B}\}}{p_B(t)}\right), \qquad t \in [0,1].\label{def:Z}
\end{equation}
Let $X$ and $\xi_O$ satisfy Assumption~\ref{ass:Y} with parameters $\alpha_X$ and $\alpha_O$, respectively. Using Assumption~\ref{ass:Missing}, it can then be shown that $Z$ also satisfies Assumption~\ref{ass:Y} with parameter $\alpha_O$. Since $X$ and~$Z$ are independent under MCAR, Lemma~\ref{lem:XZ} implies that $Y$ satisfies the conditions of Theorem~\ref{thm:Hahn} with $\gamma_1 = \min\{\alpha_X/2, \alpha_O/2\}$ and $\gamma_2 = \min\{\alpha_X, \alpha_O\}$. Consequently, $\sqrt{n} (\widehat{\mu}_A -\widehat{\mu}_B)$ converges weakly in $D[0,1]$ to a Gaussian process with covariance~\eqref{eq:CovZ}. Since the supremum is a continuous functional on $D[0,1]$, the continuous mapping theorem finishes the proof.
\end{proof}

\begin{lemma}\label{lem:hatA}
Suppose that the assumptions of Theorem~\ref{thm:CLT_D} are satisfied. Then,
\begin{enumerate}[label=(\roman*)]
\item $\norm{\widehat{\mu}_{A}-\widehat{\mu}_{\widehat{A}}}_\infty = o_p(1)$,
\item if MCAR is true, $\norm{\widehat{\mu}_{A}-\widehat{\mu}_{\widehat{A}}}_\infty = o_p(1/\sqrt{n})$.
\end{enumerate}
\end{lemma}

\begin{proof}
We can rewrite the difference of empirical means as
\begin{align*}
&\widehat{\mu}_{A}(t)-\widehat{\mu}_{\widehat{A}}(t)\\ &\qquad =\frac{\frac{1}{n}\sum_{i = 1}^n (X_i(t)-\mu(t)) \xi_{O_i}(t) \ind\{O_i \in \mathcal{A}\}}{\widehat{p}_A(t)}-\frac{\frac{1}{n}\sum_{i = 1}^n (X_i(t)-\mu(t)) \xi_{O_i}(t) \ind\{O_i \in \widehat{\mathcal{A}}\}}{\widehat{p}_{\widehat{A}}(t)},
\end{align*}
where $\mu(t) = \expv[X(t)]$. To simplify the notation, set $\delta_i = \ind\{O_i \in \mathcal{A}\}-\ind\{O_i \in \widehat{\mathcal{A}}\}$ and define
\begin{align*}
S_n(t) = \frac{1}{n}\sum_{i = 1}^n (X_i(t)-\mu(t)) \xi_{O_i}(t) \ind\{O_i \in \mathcal{A}\}, && S_n'(t) = \frac{1}{n}\sum_{i = 1}^n (X_i(t)-\mu(t)) \xi_{O_i}(t) \ind\{O_i \in \widehat{\mathcal{A}}\}.
\end{align*}
Then,
\begin{equation}\label{eq:hatmuA-hatmuAhat}
\widehat{\mu}_{A}(t)-\widehat{\mu}_{\widehat{A}}(t) = \frac{S_n(t) (\widehat{p}_{\widehat{A}}(t)-\widehat{p}_{A}(t)) + (S_n(t)-S_n'(t))\widehat{p}_{A}(t)}{\widehat{p}_{\widehat{A}}(t)\widehat{p}_{A}(t)}.
\end{equation}
Note that
\begin{equation*}
\expv\left[ \norm{\widehat{p}_A - \widehat{p}_{\widehat{A}}}_{\infty}\right] \leq \frac{1}{n} \sum_{i=1}^n \expv\left[\norm{\xi_{O_i} \delta_i}_\infty\right] =\prob(O_1 \in \widehat{\mathcal{A}}  \, \sym \,  \mathcal{A}) = o(1).
\end{equation*}
Together with \eqref{eq:LLN_O} and Assumption~\ref{ass:Missing}, it then follows that the denominator in \eqref{eq:hatmuA-hatmuAhat} is asymptotically bounded from below.

For (i), it thus suffices to note that $\norm{S_n}_\infty = O_p(1)$ and
\begin{equation*}
    \expv\left[\norm{S_n-S_n'}_\infty\right] = \expv\left[\norm{\frac{1}{n}\sum_{i=1}^n (X_i-\mu) \xi_{O_i} \delta_i}_\infty \right] \leq \expv \left[\norm{X-\mu}^2_\infty\right]^{1/2} \prob(O_1 \in \widehat{\mathcal{A}}  \, \sym \,  \mathcal{A})^{1/2} = o(1),
\end{equation*}
where we have used the Cauchy-Schwarz inequality and the assumptions of the lemma.

For (ii), an argument based on Hahn's central limit theorem yields $\sqrt{n}\,\norm{S_n}_\infty = O_p(1)$. Invoking~\eqref{eq:hatmuA-hatmuAhat}, it suffices to show that $\sqrt{n} \,\norm{S_n-S_n'}_\infty = \norm{\frac{1}{\sqrt{n}}\sum_{i = 1}^n (X_i-\mu) \xi_{O_i}\delta_i}_\infty = o_p(1)$. To this end, we invoke Theorem~13.5 in \cite{BillingsleyPatrick1979Pam} and verify the conditions (13.11) as well as (13.14) of the theorem to show that $\sqrt{n}(S_n-S_n')$ converges weakly to the zero function. Set $Y = (X-\mu) \xi_O$, and note that for $0 \leq t_1 \leq t_2 \leq \dots \leq t_K \leq 1$, the Markov inequality implies
\begin{equation*}
\prob\left( \sum_{k=1}^K \abs{\frac{1}{\sqrt{n}}\sum_{i = 1}^n Y_i(t_k)\delta_i} > \varepsilon \right)  \leq \frac{K}{\varepsilon^2n}\sum_{k=1}^K \expv \left[ \left(\sum_{i = 1}^n (X_i(t_k)-\mu(t_k)) \xi_{O_i}(t_k)\delta_i\right)^2\right].
\end{equation*}
An argument based on conditioning shows that the mixed terms in the above sum have zero expectation under MCAR. Therefore, we obtain for any $k \leq K$,
\begin{align*}
\frac{1}{n}\,\expv \left[ \left(\sum_{i = 1}^n (X_i(t_k)-\mu(t_k)) \xi_{O_i}(t_k)\delta_i\right)^2\right] &\leq \expv\left[(X(t_k)-\mu(t_k))^2 \right] \expv\left[\abs{\delta_1} \right] \\
&= \expv\left[(X(t_k)-\mu(t_k))^2 \right] \prob(O_1 \in \widehat{\mathcal{A}}  \, \sym \,  \mathcal{A}) = o(1),
\end{align*}
where we again used independence between $(X_i)_{i=1}^n$ and $(O_i)_{i=1}^n$ under MCAR. Combining the above, shows the condition (13.11). As for the condition~(13.14),  note that for $0 \leq s \leq t\leq u \leq 1$,
\begin{align*}
&\expv \left[ \left(\frac{1}{\sqrt{n}}\sum_{i = 1}^n (Y_i(s)-Y_i(t))\delta_i\right)^2\left(\frac{1}{\sqrt{n}}\sum_{i = 1}^n (Y_i(t)-Y_i(u))\delta_i\right)^2\right] \\
&\qquad \leq \expv\left[(Y(s)-Y(t))^2(Y(t)-Y(u))^2\right] + 3 \,\expv\left[(Y(s)-Y(t))^2\right]\expv\left[(Y(t)-Y(u))^2\right]\\
&\qquad \leq C_0 \abs{u-s}^\beta
\end{align*}
with $\beta = \min\{\alpha_X, \alpha_O\} > 1$ where we have used Lemma~\ref{lem:XZ} and the fact that $X$ and $\xi_O$ satisfy Assumption~\ref{ass:Y}. This shows condition (13.14) and Theorem~13.5 in \cite{BillingsleyPatrick1979Pam} then finishes the proof of the lemma.
\end{proof}

\subsection{Proof of Theorem~\ref{thm:CLT_mu}}

Let $H = L^2([0,1]\times \Xi, \lambda \otimes \nu)$ be the space of two-dimensional functions which are square integrable with respect to the product measure~$\lambda \otimes \nu$, where $\lambda$ denotes the Lebesgue measure on~$[0,1]$. For $f \in H$, define $\norm{f}^2_H = \int_{[0,1]\times \Xi} f(t,z)^2 \,\dF(t,z)$, and for $f \in L^2[0,1]$, set $\norm{f}_{L^2}^2 = \int_0^1 f(t)^2 \, \dt$.

\begin{proof}[of Theorem~\ref{thm:CLT_mu}]
Similar to Lemma~\ref{lem:hatA}, it can be shown that $\sqrt{n}\,\norm{ \widehat{F}_{\widehat{A}}- \widehat{F}_{A}}_H + \sqrt{n}\,\norm{ \widehat{F}_{\widehat{B}}- \widehat{F}_{B}}_H = o_p(1)$. It thus suffices to derive the limiting distribution of $\sqrt{n}\,\norm{ \widehat{F}_{A}- \widehat{F}_{B}}_H$. To this end, note that $\expv\left[(\ind\{X_i(t) \leq z\}-F_A(t,z)) \xi_{O_i}(t) \ind\{O_i \in \mathcal{A}\} \right] = 0$. The central limit theorem of Hilbert-valued random variables (Theorem~2.7 in \cite{bosq2000linear}) implies
\begin{equation}\label{eq:LLN_H}
\frac{1}{n} \sum_{i=1}^n (\ind\{X_i(t) \leq z\}-F_A(t,z))  \xi_{O_i}(t)\ind\{O_i \in \mathcal{A}\}  = O_p(1/\sqrt{n}),
\end{equation}
in the norm $\norm{\cdot}_H$. 
From \eqref{eq:LLN_O} and \eqref{eq:LLN_H}, it follows that
\begin{equation*}
\sqrt{n} \, (\widehat{F}_A(t,z) - F_A(t,z)) = \frac{1}{\sqrt{n}} \sum_{i = 1}^n \frac{(\ind\{X_i(t) \leq z\}-F_A(t,z)) \xi_{O_i}(t) \ind\{O_i \in \mathcal{A}\}}{p_A(t)} + o_p(1).
\end{equation*}
Analogous results hold if $\mathcal{A}$ is replaced with $\mathcal{B}$. Now, since $\FA = \FB = F$ under MCAR,
\begin{align*}
&\sqrt{n} \, (\widehat{F}_A(t,z) - \widehat{F}_B(t,z)) \\
&\qquad = \frac{1}{\sqrt{n}} \sum_{i = 1}^n (\ind\{X_i(t) \leq z\}-F(t,z))\xi_{O_i}(t)\left(\frac{\ind\{O_i \in \mathcal{A}\}}{p_A(t)} - \frac{\ind\{O_i \in \mathcal{B}\}}{p_B(t)}\right) + o_p(1).
\end{align*}
Another application of the central limit theorem finishes the proof of the theorem.
\end{proof}

\subsection{Proofs of results in Section~\ref{sec:Consistency}}

\begin{proof}[of Theorem~\ref{thm:power}]
We only show (i). By the inverse triangle inequality,
\begin{align*}
T_\mu &= \sqrt{n} \, \norm{\hatmA-\hatmB}_\infty \\
&\geq \sqrt{n} \left(\norm{\mu_A-\mu_B}_\infty - \norm{(\widehat{\mu}_A-\widehat{\mu}_B)-(\mu_A-\mu_B)}_\infty - \norm{\hatmA - \widehat{\mu}_A}_\infty -  \norm{\hatmB - \widehat{\mu}_B}_\infty \right).
\end{align*}
Lemma~\ref{lem:hatA}(i) shows that $\norm{\hatmA - \widehat{\mu}_A}_\infty +  \norm{\hatmB - \widehat{\mu}_B}_\infty = o_p(1)$. Using the law of large numbers in the Skorokhod space $D[0,1]$ \citep[Theorem~1]{rao1963law}, we get $\norm{(\widehat{\mu}_A-\widehat{\mu}_B)-(\mu_A-\mu_B)}_\infty = o_p(1)$. The assumption $\norm{\mu_A-\mu_B}_\infty > 0$ then implies $T_\mu \xrightarrow{p} \infty$ as $n \to \infty$. 
\end{proof}

\noindent Consider again the mean-based test statistic $T_\mu$. In case of a fixed clustering rule given by $\widehat{\mathcal{A}} = \mathcal{A}$ and $\widehat{\mathcal{B}} = \mathcal{B}$, the quantities $\norm{\hatmA - \widehat{\mu}_A}_\infty$ and $\norm{\hatmB - \widehat{\mu}_B}_\infty$ vanish. Under the additional assumption that the product $X\xi_{O}$ satisfies Assumption~\ref{ass:Y}, an argument based on the CLT then shows that $\norm{(\widehat{\mu}_A-\widehat{\mu}_B)-(\mu_A-\mu_B)}_\infty = O_p(1/\sqrt{n})$. This implies that $T_\mu$ has power against local alternatives of the form $\norm{\mu_{A,n}-\mu_{B,n}}_\infty = \frac{c}{\sqrt{n}}$ for  $c > 0$.

\begin{proof}[of Proposition~\ref{prop:ConsistencyDist}]
    Due to continuity in (i), it suffices to show that $F_A \neq F_B$ in the single point~$(t_0, z_0)$. For a contrapositive proof, assume that $F_A(t_0, z_0) = F_B(t_0, z_0)$ which  in turn implies $\prob(X(t_0) \leq z_0\vert t_0 \in O) = F_A(t_0,z_0)$. Using the definition of $F_A(t,z)$, we get
    \begin{align*}
    \prob(X(t_0) \leq z_0 \vert t_0 \in O) = \prob(X(t_0) \leq z_0 \vert t_0 \in O, O \in \mathcal{A}).
    \end{align*}
    This means $\{X(t_0) \leq z_0\} \perp\!\!\!\!\perp \{O \in \mathcal{A}\} \vert \{t_0 \in O\}$.
\end{proof}

\begin{proof}[of Proposition~\ref{prop:consistency}]
For a contrapositive proof, assume that
\begin{equation*}
\norm{\FA-\FB}_H^2 = \int_0^1 \int_{\Xi}(\FA(t,z)-\FB(t,z))^2 \dnu(z)\,\dt = 0.
\end{equation*}
Hence, $\int_{\Xi}(\FA(t,z)-\FB(t,z))^2 \dnu(z) = 0$ for almost all $t\in[0,1]$. Because $z \mapsto \FA(t,z)-\FB(t,z)$ is right-continuous and $\nu$ is supported on $\Xi$, we get $\FA(t,z) = \FB(t,z)$ for all $z \in \textup{int}(\Xi)$. Since it holds $X(t) \in \textup{int}(\Xi)$ almost surely, the equality of distributions implies $\mA(t)=\mB(t)$. This is true for almost all $t \in [0,1]$ and therefore $\norm{\mA-\mB}_{\infty}=0$ by continuity of $t \mapsto \mu_A(t)-\mu_B(t)$.
\end{proof}

For the proof of Proposition~\ref{prop:BrownianMean}, we define the stopping time
\begin{equation*}
\tau = \tau_{ab} = \inf\{t \geq 0\colon B(t) \in \{a,b\}\}.
\end{equation*}
In addition, we denote by $\Phi$ and $\varphi$ the cdf and pdf of a standard normal distribution, respectively. The proof of Proposition~\ref{prop:BrownianMean} is then based on the following lemmas.

\begin{lemma}\label{lem:Etau}
For $a < 0 < b$, it holds that
\begin{equation*}
\expv[B(1) \vert \tau \geq 1] = \expv[B(\tau) \vert \tau \geq 1].
\end{equation*}
\end{lemma}

\begin{proof}
Let $\mathcal{F}_\tau$ be the $\sigma$-algebra of events observable at time $\tau$. By $\expv[B(1)] = 0$ and the strong Markov property, we get
\begin{align*}
\expv[B(1) \ind \{\tau \geq 1\}] &= - \expv[B(1) \ind\{\tau < 1\}] \\
&= - \expv[B(\tau) \ind\{\tau < 1\}] - \expv[(B(1)-B(\tau)) \ind\{\tau < 1\}] \\
&= - \expv[B(\tau) \ind\{\tau < 1\}] - \expv[\expv[B(1)-B(\tau)\vert \mathcal{F}_\tau] \ind\{\tau < 1\}]\\
&= - \expv[B(\tau) \ind\{\tau < 1\}] + 0.
\end{align*}
The result of the proposition then follows from $\expv[B(\tau)] = 0$ as a consequence of optimal stopping, and the definition of conditional expectations.
\end{proof}

\begin{lemma}\label{lem:probs}
For $a < 0 < b$, it holds that
\begin{align*}
&\prob(B(\tau) = a, \tau \geq 1) = -2\sum_{k=1}^\infty \left(\Phi(a+2k(b-a))-\Phi(-a + 2(k-1)(b-a))\right) + \frac{b}{b-a},\\
&\prob(B(\tau) = b, \tau \geq 1) = -2\sum_{k=1}^\infty \left(\Phi(-b+2k(b-a))-\Phi(b + 2(k-1)(b-a))\right) + \frac{-a}{b-a}.
\end{align*}
\end{lemma}

\begin{proof}
The assertion follows from integrating the formulas (3.0.6)(a) and (3.0.6)(b) given in \cite{Borodin1996handbook}.
\end{proof}

\begin{proof}[of Proposition~\ref{prop:BrownianMean}]

By definition, we have that
\begin{equation*}
\mu_A(t) = \expv[B(t) \vert B(s) \in (a,b) \,\forall s], \qquad \mu_B(t) = \expv[B(t) \vert B(t) \in (a,b), \exists s\, B(s) \not\in (a,b)].
\end{equation*}
Note that by continuity of $t \mapsto \mA(t)-\mB(t)$, it actually suffices to show that~$\mA(1) \neq \mB(1)$. Let us also define $\mu_C(t) = \expv[B(t)\vert B(t) \in (a,b)]$ and note that it equals
\begin{equation*}
\mu_A(t) \prob(B(s) \in (a,b)\, \forall s \vert B(t) \in (a,b)) + \mu_B(t) (1-\prob(B(s) \in (a,b)\, \forall s \vert B(t) \in (a,b))).
\end{equation*}
Hence, $\mu_C(1)$ is a convex combination of $\mu_A(1)$ and $\mu_B(1)$ with nonzero weights. Therefore, $\mu_A(1) \neq \mu_B(1)$ if and only if $\mu_A(1) \neq \mu_C(1)$ which we show in the following. The conditional mean $\mA$ can be written as
\begin{equation*}
\mA(t) = \expv[B(t) \vert B(s) \in (a,b) \,\forall s] = \expv[B(t)\vert \tau \geq 1].
\end{equation*}
Due to Lemma~\ref{lem:Etau}, this implies $\mA(1) = \expv[B(\tau)\vert \tau \geq 1]$. Furthermore, note that $\mu_C$ relates to the mean of a truncated normal distribution. It can easily be shown that
\begin{equation*}
\mu_C(1) = \expv[B(1)\vert B(1) \in (a,b)] = \frac{\varphi(a)-\varphi(b)}{\Phi(b)+\Phi(-a)-1}.
\end{equation*}
We need to prove
\begin{equation*}
\expv[B(\tau)\vert \tau \geq 1] \neq \frac{\varphi(a)-\varphi(b)}{\Phi(b)+\Phi(-a)-1}.
\end{equation*}
Set $a = -b + \varepsilon$ for $b$ large and $\varepsilon \in (0,1)$ sufficiently small. To prove the above inequality, we will show that
\begin{equation}\label{eq:Inequality}
\expv[B(\tau)\ind\{\tau\geq 1\}] > \frac{\varphi(a)-\varphi(b)}{\Phi(b)+\Phi(-a)-1}.
\end{equation}
\begin{itemize}
\item Regarding the lhs of \eqref{eq:Inequality}, Lemma~\ref{lem:probs} implies,
\begin{align*}
&\expv[B(\tau)\ind\{\tau\geq 1\}]\\
 &\quad = a \prob(B(\tau)=a, \tau \geq 1) + b \prob(B(\tau)=b, \tau \geq 1) \\
&\quad = -2 \sum_{k=1}^\infty \Big\lbrace a(\Phi(a+2k(b-a))-\Phi(-a + 2(k-1)(b-a))) \\
&\qquad \qquad + b( \Phi(-b+2k(b-a))-\Phi(b + 2(k-1)(b-a))) \Big\rbrace.
\end{align*}
For $a = -b + \varepsilon$, we get
\begin{align*}
\expv[B(\tau)\ind\{\tau\geq 1\}] &>2b(\Phi(3b-\varepsilon)-\Phi(b-\varepsilon)-\Phi(3b-2\varepsilon)+\Phi(b))\\
&\qquad -2\varepsilon\sum_{k=1}^\infty\left\lbrace \Phi(-b+\varepsilon+2k(2b-\varepsilon))-\Phi(b-\varepsilon+2(k-1)(2b-\varepsilon))\right\rbrace .
\end{align*}
A Taylor expansion yields
\begin{align*}
\expv[B(\tau)\ind\{\tau\geq 1\}]&> 2\varepsilon b(\varphi(b)+\varphi(3b)) - 2\varepsilon\sum_{k=1}^\infty\left\lbrace \Phi(-b+4kb)-\Phi(b+4(k-1)b)\right\rbrace + o(\varepsilon)\\
&> 2\varepsilon\left(b\varphi(b)-\sum_{k=1}^\infty\left\lbrace \Phi(-b+4kb)-\Phi(b+4(k-1)b)\right\rbrace\right)  + o(\varepsilon).
\end{align*}
\item Regarding the rhs of~\eqref{eq:Inequality}, let $a = -b + \varepsilon$ and note that $\Phi(b)+\Phi(b-\varepsilon)\to 2$ as $b \to \infty$. Therefore,
\begin{equation*}
\frac{\varphi(b-\varepsilon)-\varphi(b)}{\Phi(b)+\Phi(b-\varepsilon)-1} < \frac{3}{2}\,\left( \varphi(b-\varepsilon)-\varphi(b)\right),
\end{equation*}
provided that $b$ is sufficiently large. Another Taylor expansion yields
\begin{equation*}
\frac{\varphi(b-\varepsilon)-\varphi(b)}{\Phi(b)+\Phi(b-\varepsilon)-1} < \frac{3}{2} \, \varepsilon b\varphi(b) + o(\varepsilon).
\end{equation*}
\end{itemize}
To show~\eqref{eq:Inequality}, we thus need to prove
\begin{align*}
&2\varepsilon\left(b\varphi(b)-\sum_{k=1}^\infty\left\lbrace \Phi(-b+4kb)-\Phi(b+4(k-1)b)\right\rbrace\right) - \frac{3}{2} \, \varepsilon b\varphi(b) +o(\varepsilon) \\
&\qquad = 2\varepsilon b \varphi(b)\left(\frac{1}{4} - \frac{\sum_{k=1}^\infty\left\lbrace \Phi(-b+4kb)-\Phi(b+4(k-1)b)\right\rbrace}{b\varphi(b)}\right) + o(\varepsilon) > 0.
\end{align*}
Since $\frac{\sum_{k=1}^\infty\left\lbrace \Phi(-b+4kb)-\Phi(b+4(k-1)b)\right\rbrace}{b\varphi(b)} \leq \frac{1-\Phi(b)}{b\varphi(b)}  \to 0$ as $b \to \infty$, the above inequality is true provided that $b$ is sufficiently large. This finishes the proof of Proposition~\ref{prop:BrownianMean}.
\end{proof}

\section{Computational algorithms}\label{sec:ComputationalAlgorithms}

In this section, we provide more details on the practical implementation of our methodology. For testing MCAR, we discuss a procedure based on asymptotic distributions (Section~\ref{sec:Approximationsbased}) as well as a bootstrap (Section~\ref{sec:Bootstrap}). The benchmark method of \cite{liebl2019partially} was implemented using their R-package \citep{RamesederFD2017PartiallyFD}. Random functions were expanded on a common domain using a Fourier basis of $J = 11$ functions; see their original paper for further details.

\subsection{Approximations based on asymptotic distributions}\label{sec:Approximationsbased}

The covariance function $k(s,t)$ in \eqref{eq:CovZ} is estimated by
\begin{equation*}
\widehat{k}(s,t) = \frac{1}{n}\sum_{i=1}^n\left(\frac{ \tilde{X}_i(s)\tilde{X}_i(t)\xi_{O_i}(s)\xi_{O_i}(t)\ind\{O_i \in \widehat{\mathcal{A}}\}}{\widehat{p}_{\widehat{A}}(s)\widehat{p}_{\widehat{A}}(t)} + \frac{ \tilde{X}_i(s)\tilde{X}_i(t)\xi_{O_i}(s)\xi_{O_i}(t)\ind\{O_i \in \widehat{\mathcal{B}}\}}{\widehat{p}_{\widehat{B}}(s)\widehat{p}_{\widehat{B}}(t)} \right),
\end{equation*}
where $\tilde{X}_i(t) = X_i(t) - (\hatmA(t)\ind\{O_i\in\widehat{\mathcal{A}}\} + \hatmB(t)\ind\{O_i\in\widehat{\mathcal{B}}\})$. This yields Algorithm~\ref{alg:D} and~\ref{alg:SCB}.

\vspace{-0.1cm}

\begin{algorithm}[!h]
\caption{Asymptotic $T_{\mu}$ test}\label{alg}\label{alg:D}
\nl Compute the test statistic $T_\mu$ from $(X_i, O_i)_{i=1}^n$\;

\nl Estimate the covariance function $k$ by $\widehat{k}$\;
\nl Compute the first $q$ eigenvalues $(\widehat{\lambda}_j)_{j=1}^q$  and eigenfunctions $(\widehat{\phi}_j)_{j=1}^q$  of $\widehat{k}$\;
\nl \For{$b = 1, \dots, B^*$}{
Sample $Z_1^{(b)}, Z_2^{(b)}, \dots, Z_q^{(b)}$ i.i.d.~$N(0,1)$ random variables\;
Compute $W^{(b)} = \sup_{t \in [0,1]} \abs{ \sum_{j=1}^qZ_j^{(b)} (\widehat{\lambda}_j)^{1/2}\widehat{\phi}_j(t)}$\;}
\nl Approximate the $p$-value using $T_\mu$ and $W^{(1)}, \dots, W^{(B^*)}$.
\end{algorithm}

\begin{algorithm}[!h]
\caption{Simultaneous confidence bands for $\hatmA-\hatmB$}\label{alg:SCB}

\nl Compute $\hatmA$ and $\hatmB$ from $(X_i, O_i)_{i=1}^n$\;
\nl Estimate the covariance function $k$ by $\widehat{k}$\;
\nl Compute the first $q$ eigenvalues $(\widehat{\lambda}_j)_{j=1}^q$  and eigenfunctions $(\widehat{\phi}_j)_{j=1}^q$  of $\widehat{k}$\;
\nl \For{$b = 1, \dots, B^*$}{
Sample $Z_1^{(b)}, Z_2^{(b)}, \dots, Z_q^{(b)}$ i.i.d.~$N(0,1)$ random variables\;
Compute $W^{(b)} = \sup_{t \in [0,1]} \abs{ \sum_{j=1}^q Z_j^{(b)}(\widehat{\lambda}_j)^{1/2}\widehat{\phi}_j(t)}$\;}
\nl Compute the empirical $(1-\alpha)$-quantile $\widehat{q}_\alpha$ of $W^{(1)}, \dots, W^{(B^*)}$\;
\nl Compute the simultaneous confidence band $\hatmA - \hatmB \pm \widehat{q}_\alpha/\sqrt{n}$.

\end{algorithm}

\noindent The truncation parameter $q$ is chosen according to the fraction of variance explained (FVE) criterion such that $\sum_{j=1}^q \widehat{\lambda}_j / \sum_{j} \widehat{\lambda}_j \geq 0.99$.

To approximate the limit distribution in Theorem \ref{thm:CLT_mu}, it is useful to consider the Karhunen-Loève expansion of the two-dimensional Gaussian random function $Z$,
\begin{equation*}
Z(t,z) = \sum_{j=1}^\infty \zeta_j \psi_j(t,z),\qquad t \in [0,1], z \in \Xi. 
\end{equation*}
Here, $(\psi_j)_{j=1}^\infty$ refers to the eigenfunctions of the integral operator on $H = L^2([0,1]\times \Xi, \lambda \otimes \nu)$ with kernel $\rho(s,t,z_1,z_2) = \cov(Z(s,z_1), Z(t,z_2))$. The eigenfunctions can be completed to a basis and satisfy $\langle \psi_j, \psi_k\rangle_H = \delta_{jk}$ where $ \langle f, g\rangle_H = \int_0^1 \int_\Xi f(t,z)g(t,z) \,\text{d}(\lambda \otimes \nu)(t,z)$ for elements~$f,g \in H$. Furthermore, the coefficients $(\zeta_j)_{j=1}^\infty$ are independent Gaussian random variables given by the projections $\zeta_j = \langle Z, \psi_j\rangle_H$. The $j$-th coefficient $\zeta_j$ has mean zero and variance $\kappa_j$ where $(\kappa_j)_{j=1}^\infty$ refers to the sequence of eigenvalues of the corresponding integral operator. See also \cite{zhou2014principal} for further details on principal components analysis of two-dimensional functional data. We obtain
\begin{equation*}
\int_{[0,1]\times \Xi} Z(t, z)^2 \, \dF(t,z) \quad \overset{d}{=} \quad \sum_{j=1}^\infty \kappa_j \chi_{1;j}^2,
\end{equation*}
for a sequence $(\chi_{1;j}^2)_{j=1}^\infty$ of independent chi-square distributed random variables. To estimate the eigenvalues, we draw Monte Carlo points $\{(t_k, z_k)\colon k = 1, \dots, K\}$ from the  distribution $\lambda \otimes \nu$. We then estimate $(\kappa_j)_{j=1}^q$ by the eigenvalues $(\widehat{\kappa}_j)_{j=1}^q$ of the matrixized covariance estimator,
\begin{align*}
\widehat{\rho}(t_k, t_\ell, z_k, z_{\ell})  = \frac{1}{n}\sum_{i=1}^n&\Bigg(\frac{ \tilde{\ind}_i(t_k,z_k)\tilde{\ind}_i(t_\ell,z_\ell)\xi_{O_i}(t_k)\xi_{O_i}(t_\ell)\ind\{O_i \in \widehat{\mathcal{A}}\}}{\widehat{p}_{\widehat{A}}(t_k)\widehat{p}_{\widehat{A}}(t_\ell)}\\
&  \qquad + \frac{ \tilde{\ind}_i(t_k,z_k)\tilde{\ind}_i(t_\ell,z_\ell)\xi_{O_i}(t_k)\xi_{O_i}(t_\ell)\ind\{O_i \in \widehat{\mathcal{B}}\}}{\widehat{p}_{\widehat{B}}(t_k)\widehat{p}_{\widehat{B}}(t_\ell)} \Bigg),
\end{align*}
where $\tilde{\ind}_i(t,z) = \ind\{X_i(t)\leq z\} - (\hatFA(t,z)\ind\{O_i\in\widehat{\mathcal{A}}\} + \hatFB(t,z)\ind\{O_i\in\widehat{\mathcal{B}}\})$. This yields Algorithm~\ref{alg:Dist}.

\begin{algorithm}[h]
\caption{Asymptotic $T_{F}$ test}\label{alg:Dist}

\nl Choose a measure $\nu$ on $\Xi \subseteq \mathbb{R}$\;
\nl Compute the test statistic $T_{F}$ from $(X_i, O_i)_{i=1}^n$\;
\nl Compute the first $q$ eigenvalues $(\widehat{\kappa}_j)_{j=1}^q$ of $\widehat{\rho}$\;
\nl \For{$b = 1, \dots, B^*$}{
Sample $Z_1^{(b)}, Z_2^{(b)}, \dots, Z_q^{(b)}$ i.i.d.~$N(0,1)$ random variables\;
Compute $W^{(b)} = \sum_{j=1}^q (Z_j^{(b)})^2\widehat{\kappa}_j$\;}
\nl Approximate the $p$-value using $T_{F}$ and $W^{(1)}, \dots, W^{(B^*)}$.

\end{algorithm}

\noindent The truncation parameter $q$ is chosen according to the fraction of variance explained (FVE) criterion such that $\sum_{j=1}^q \widehat{\kappa}_j / \sum_{j} \widehat{\kappa}_j \geq 0.99$.

\subsection{Bootstrap approximations}\label{sec:Bootstrap}

Another option to approximate the distribution of our test statistics is to consider a bootstrap procedure; see also \cite{efron1994introduction} for a general introduction. Bootstrap hypothesis tests for two samples of functional data have been considered in \cite{paparoditis2016bootstrap} and \cite{kraus2019inferential} in a setting similar to ours. In case of the mean comparison, the bootstrap resamples are usually drawn groupwise under the null hypothesis of equal means. A corresponding adaption of this approach to our case is depicted in Algorithm~\ref{alg:bootmean}. In contrast, for a comparison of distributions, the resamples are typically drawn from a pooled sample under the null hypothesis of equal distributions (see also \cite{paparoditis2016bootstrap} and \cite{kraus2019inferential}). However, a direct adaption to our case of MCAR testing seems impossible since the characteristics of observation patterns in the subgroups will generally be different under MCAR. For this reason, we also draw groupwise resamples for the distributional test and provide the details in Algorithm~\ref{alg:bootdist}. Algorithm~\ref{alg:bootSCB} finally provides an adaption of Algorithm~\ref{alg:SCB} using the bootstrap.

\begin{algorithm}[!h]
\caption{Bootstrap approximation for mean test $T_\mu$}\label{alg:bootmean}
\nl Compute $\hatmA$, $\hatmB$ and the test statistic $T_\mu$ from $(X_i, O_i)_{i=1}^n$\;
\nl Compute $n_A = \sum_{i=1}^n \ind\{O_i \in \widehat{\mathcal{A}}\}$ and $n_B = \sum_{i=1}^n \ind\{O_i \in \widehat{\mathcal{B}}\}$\;
\nl Create a data set $(X_j^{A}, O_j^{A})_{j=1}^{n_A}$ containing the $(X_i - \hatmA, O_i)$ with $O_i \in \widehat{\mathcal{A}}$\;
\nl  Create a data set $(X_j^{B}, O_j^{B})_{j=1}^{n_B}$ containing the $(X_i - \hatmB, O_i)$ with $O_i \in \widehat{\mathcal{B}}$\;
\nl \For{$b = 1, \dots, B^*$}{
Sample $(X_j^{A(b)}, O_j^{A(b)})_{j=1}^{n_A}$ from $(X_j^{A}, O_j^{A})_{j=1}^{n_A}$ with replacement\;
Sample $(X_j^{B(b)}, O_j^{B(b)})_{j=1}^{n_B}$ from $(X_j^{B}, O_j^{B})_{j=1}^{n_B}$ with replacement\;
Compute $T_\mu^{(b)}$ from $(X_j^{A(b)}, O_j^{A(b)})_{j=1}^{n_A}$ and $(X_j^{B(b)}, O_j^{B(b)})_{j=1}^{n_B}$\;}
\nl Approximate the $p$-value using $T_\mu$ and $T_\mu^{(1)}, \dots, T_\mu^{(B^*)}$.
\end{algorithm}

\begin{algorithm}[!h]
\caption{Bootstrap approximation for distribution test $T_{F}$}\label{alg:bootdist}
\nl Choose a measure $\nu$ on $\Xi \subseteq \mathbb{R}$\;
\nl Compute $\hatFA, \hatFB$ and the test statistic $T_{F}$ from $(X_i, O_i)_{i=1}^n$\;
\nl Compute $n_A = \sum_{i=1}^n \ind\{O_i \in \widehat{\mathcal{A}}\}$ and $n_B = \sum_{i=1}^n \ind\{O_i \in \widehat{\mathcal{B}}\}$\;
\nl Create a data set $(X_j^{A}, O_j^{A})_{j=1}^{n_A}$ containing the $(X_i, O_i)$ with $O_i \in \widehat{\mathcal{A}}$\;
\nl  Create a data set $(X_j^{B}, O_j^{B})_{j=1}^{n_B}$ containing the $(X_i, O_i)$ with $O_i \in \widehat{\mathcal{B}}$\;
\nl \For{$b = 1, \dots, B^*$}{
Sample $(X_j^{A(b)}, O_j^{A(b)})_{j=1}^{n_A}$ from $(X_j^{A}, O_j^{A})_{j=1}^{n_A}$ with replacement\;
Sample $(X_j^{B(b)}, O_j^{B(b)})_{j=1}^{n_B}$ from $(X_j^{B}, O_j^{B})_{j=1}^{n_B}$ with replacement\;
Compute $\hatFA^{(b)}$ and $\hatFB^{(b)}$ from $(X_j^{A(b)}, O_j^{A(b)})_{j=1}^{n_A}$ and $(X_j^{B(b)}, O_j^{B(b)})_{j=1}^{n_B}$\;
Compute $T_F^{(b)} = n \int_{\Xi}\int_0^1 ((\widehat{F}^{(b)}_A(t,z) - \hatFA(t,z)) - (\widehat{F}^{(b)}_B(t,z)-\hatFB(t,z))^2 \dt\, \dnu(z)$\;}
\nl Approximate the $p$-value using $T_F$ and $T_F^{(1)}, \dots, T_F^{(B^*)}$.
\end{algorithm}

\begin{algorithm}[!h]
\caption{Bootstrap confidence bands for $\hatmA-\hatmB$}\label{alg:bootSCB}
\nl Compute $n_A = \sum_{i=1}^n \ind\{O_i \in \widehat{\mathcal{A}}\}$ and $n_B = \sum_{i=1}^n \ind\{O_i \in \widehat{\mathcal{B}}\}$\;
\nl Create a data set $(X_j^{A}, O_j^{A})_{j=1}^{n_A}$ containing the $(X_i - \hatmA, O_i)$ with $O_i \in \widehat{\mathcal{A}}$\;
\nl  Create a data set $(X_j^{B}, O_j^{B})_{j=1}^{n_B}$ containing the $(X_i - \hatmB, O_i)$ with $O_i \in \widehat{\mathcal{B}}$\;
\nl \For{$b = 1, \dots, B^*$}{
Sample $(X_j^{A(b)}, O_j^{A(b)})_{j=1}^{n_A}$ from $(X_j^{A}, O_j^{A})_{j=1}^{n_A}$ with replacement\;
Sample $(X_j^{B(b)}, O_j^{B(b)})_{j=1}^{n_B}$ from $(X_j^{B}, O_j^{B})_{j=1}^{n_B}$ with replacement\;
Compute $T_{\mu}^{(b)}$ from $(X_j^{A(b)}, O_j^{A(b)})_{j=1}^{n_A}$ and $(X_j^{B(b)}, O_j^{B(b)})_{j=1}^{n_B}$\;}
\nl Compute the empirical $(1-\alpha)$-quantile $\widehat{q}_\alpha$ of $T_{\mu}^{(1)}, \dots, T_{\mu}^{(B^*)}$\;
\nl Compute the simultaneous confidence band $\hatmA - \hatmB \pm \widehat{q}_\alpha/\sqrt{n}$.
\end{algorithm}

\clearpage

\subsection{Additional results for real data examples}

The table below provides additional information for the real data sets. Here, $B^*$ denotes the number of simulated examples used to approximate the limit distributions, and, respectively, the number of bootstrap replicates. The values of $q_\mu$ and $q_F$ report the truncation parameters which were chosen by the FVE criterion as explained in Section~\ref{sec:Approximationsbased} and Section~\ref{sec:Bootstrap}.

\begin{center}
\begin{tabular}{|c|crrrrrrr|}
\hline
data        &  distribution  & $p_{\mu}$ & $p_F$ & $T_\mu$ &  $T_F$ & $B^*$ & $q_\mu$ & $q_F$\\
\hline
heart rate  &  asymptotic & 0.765   & 0.911     & 42.49 &  0.97       & 1e6   & 10 & 60\\
electricity &  asymptotic & $<0.001$ & $<0.001$ & 47.76  & 82.40   & 1e6   & 4  & 109\\
temperature &  bootstrap  &  $0.036$ & $0.054$  & 32.59  & 2.10     & 1e6   & -   &  -\\
\hline
\end{tabular}
\end{center}

\section{Additional simulations}\label{sec:AdditionalSimulations}

In this section, we check the numerical properties of our test procedures when applied to non-Gaussian synthetic data. To this end, we repeatedly draw $n$ independent copies $X_1, \dots, X_n$ from the subsample of completely observed heart rate profiles discussed in Section~\ref{sec:HeartRateData}. The error probabilities are then estimated over 1,000 simulation runs.

\subsection{Missing completely at random}

We generate pseudo-missing data as follows:~Let $D\sim U[1/4, 1]$ and set~$O = [0, D]$. In this case, all curves are observed on $[0, 1/4]$ but only a negligible fraction covers $[0,1-\varepsilon]$ for small $\varepsilon > 0$. Since $D$ is chosen independently of $X$, MCAR holds. The table below reports the estimated type~I error probabilities for a nominal level of $\alpha = 0.05$.

\begin{table}[h]
\centering
\begin{tabular}{|c|cc|cc|}
\cline{2-5}
\multicolumn{1}{c|}{} & \multicolumn{4}{c|}{data-driven clustering}  \\
\cline{2-5}
\multicolumn{1}{c|}{} & \multicolumn{2}{c|}{asymp} & \multicolumn{2}{c|}{boot}  \\
\hline
$n$ & $T_\mu$ & $T_{F}$ & $T_\mu$ & $T_{F}$  \\
\hline
100 & 0.074 & 0.058 & 0.057 & 0.061  \\
250 & 0.062 & 0.050 & 0.054 & 0.052  \\
500 & 0.065 & 0.050 & 0.060 & 0.054  \\
\hline
\end{tabular}
\caption{Simulation results under MCAR.}
\end{table}

\subsection{Violation of MCAR}

Let $\mu(t) = \expv[X(t)]$ and define
\begin{align*}
O = \begin{cases}
[0,1], \qquad & \text{with prob.~$1-p$}, \\
[0,D],\qquad & \text{with prob.~$p$,}
\end{cases}
\end{align*}
with $D = 1$ if $X(0)-\mu(0) \geq 10$ and $D \sim U[1/4, 1]$ else. A fraction $p \in \{0.5, 0.6, \dots, 1\}$ of the curves is subject to this data-dependent missingness mechanism and MCAR is thus violated. A larger number of curves is censored if $p = 1$ and we therefore expect our tests to have higher power in this case. The intuition is confirmed by the following figure which plots the estimated rejection probabilities using our tests based on asymptotic distributions. In this scenario, the conditional mean functions of complete and incomplete curves are related by an approximate vertical shift. This favors the mean test $T_\mu$ which achieves performance comparable to the distribution test.

\begin{figure}[!h]
\centering
\includegraphics[width=0.9\textwidth]{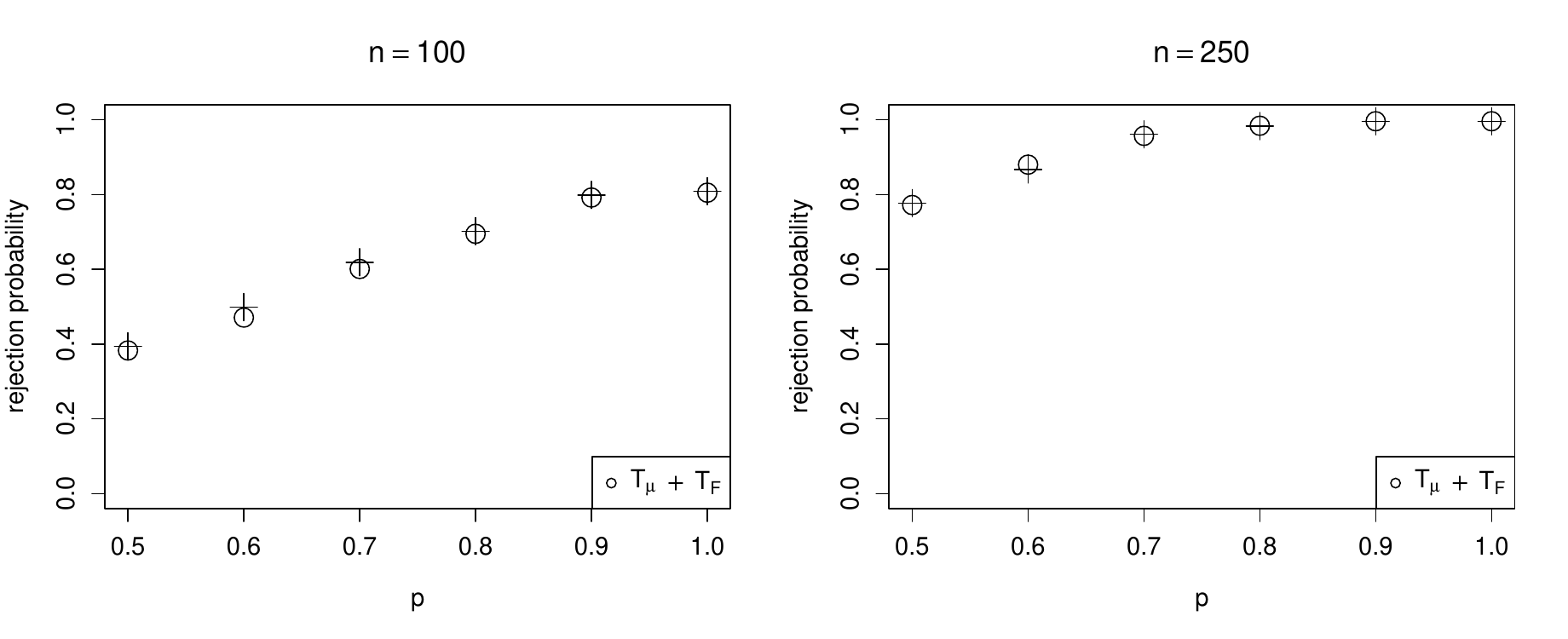}
\caption{Rejection probabilities using tests based on asymptotic distributions.}
\end{figure}

\newpage

\section{Comparison of hypotheses}

To gain further insights, we compare the two main hypotheses:
\begin{align*}
H_0\colon  & \quad \text{$X$ and $O$ are independent,} \\
H_0^F\colon & \quad  P(X(t) \leq z \vert O \in \mathcal{A}) = P(X(t) \leq z \vert O \in \mathcal{B}) \text{ for all $z \in \mathbb{R}, t \in [0,1]$}.
\end{align*}
Clearly, $H_0$ implies $H_0^F$ for any choice of the partition $\mathcal{K} = \mathcal{A} \mathbin{\dot{\cup}} \mathcal{B}$. The following examples show that the reverse implication is true under specific conditions but not in general.

\subsection{Example: $H_0$ is implied by $H_0^F$}

Let $Y$ and $W$ be two independent random functions which are independent of some observation set $O$. Define $$X(t) = Y(t) + \delta \, W(t)\ind\{O \in \mathcal{A}\}, \qquad t \in [0,1],$$ where $\delta \in \mathbb{R}$. Here, $X$ relates to a contaminated version of $Y$ and the parameter $\delta$ measures the degree of contamination. If $H_0^F$ holds, then it follows in particular that
$\var(X(t)\vert O \in \mathcal{A}) = \var(X(t)\vert O \in \mathcal{B})$ which entails $$\var(Y(t)) + \delta^2 \var(W(t)) = \var(Y(t)),$$ and thus $\delta = 0$. Since $\delta = 0$, $X = Y$ which is independent of $O$, so $H_0$ is satisfied.

\subsection{Example: $H_0$ is false but $H_0^F$ true}

Let $\zeta_1$ and $\zeta_2$ be two independent random variables with standard normal distribution. Define
\begin{align*}
X(t) = \begin{cases}
\zeta_1, \qquad & \text{for $t \in [0, 0.5)$}, \\
\zeta_2, \qquad & \text{for $t \in [0.5, 1]$,}
\end{cases} &&
O = \begin{cases}
[0,1], \qquad & \text{if $\zeta_1 \zeta_2 > 0$}, \\
[0, 0.5], \qquad & \text{otherwise.}
\end{cases}
\end{align*}
Take $t_1\in[0,0.5)$ and $t_2\in[0.5,1]$. The joint distribution of $(X(t_1),X(t_2))$ clearly depends on whether $\{\zeta_1\zeta_2>0\}$ or $\{\zeta_1\zeta_2\leq0\}$ and $H_0$ is violated. On the other hand, the marginal distributions of $X(t_1)$ and $X(t_2)$ are standard normal, both on the event $\{\zeta_1\zeta_2>0\}$ and on the event $\{\zeta_1\zeta_2\le0\}$. Therefore, $H_0^F$ holds as it is based on a comparison of univariate marginal characteristics.

\putbib[ref]
\end{bibunit}


\begin{thebibliography}{}

\bibitem[Aleksi\'c, 2024]{aleksic2024novel}
Aleksi\'c, D. (2024).
\newblock A novel test of missing completely at random: {$U$}-statistics-based
  approach.
\newblock {\em Statistics}, 58(4):1004--1023.

\bibitem[Alwan et~al., 2014]{alwan2014epidemiology}
Alwan, H., Pruijm, M., Ponte, B., Ackermann, D., Guessous, I., Ehret, G.,
  Staessen, J.~A., Asayama, K., Vuistiner, P., Younes, S.~E., et~al. (2014).
\newblock Epidemiology of masked and white-coat hypertension: the family-based
  skipogh study.
\newblock {\em PloS one}, 9(3):e92522.

\bibitem[Berrett and Samworth, 2023]{berrett2023optimal}
Berrett, T.~B. and Samworth, R.~J. (2023).
\newblock Optimal nonparametric testing of missing completely at random and its
  connections to compatibility.
\newblock {\em Ann. Statist.}, 51(5):2170--2193.

\bibitem[Billingsley, 1999]{billingsley2013convergence}
Billingsley, P. (1999).
\newblock {\em Convergence of probability measures}.
\newblock Wiley Series in Probability and Statistics: Probability and
  Statistics. John Wiley \& Sons, Inc., New York, second edition.

\bibitem[Bordino and Berrett, 2025]{bordino2024tests}
Bordino, A. and Berrett, T.~B. (2025).
\newblock Tests of missing completely at random based on sample covariance
  matrices.
\newblock {\em Ann. Statist.}, 53(5):2204--2229.

\bibitem[Bugni, 2012]{bugni2012specification}
Bugni, F.~A. (2012).
\newblock Specification test for missing functional data.
\newblock {\em Econom. Theory}, 28(5):959--1002.

\bibitem[Bugni et~al., 2009]{bugni2009goodness}
Bugni, F.~A., Hall, P., Horowitz, J.~L., and Neumann, G.~R. (2009).
\newblock Goodness-of-fit tests for functional data.
\newblock {\em Econom. J.}, 12(S1):S1--S18.

\bibitem[Deb et~al., 2020]{deb2020measuring}
Deb, N., Ghosal, P., and Sen, B. (2020).
\newblock Measuring association on topological spaces using kernels and
  geometric graphs.
\newblock {\em arXiv:2010.01768}.

\bibitem[Delaigle and Hall, 2013]{delaigle2013classification}
Delaigle, A. and Hall, P. (2013).
\newblock Classification using censored functional data.
\newblock {\em J. Amer. Statist. Assoc.}, 108(504):1269--1283.

\bibitem[Delaigle and Hall, 2016]{delaigle2016approximating}
Delaigle, A. and Hall, P. (2016).
\newblock Approximating fragmented functional data by segments of {M}arkov
  chains.
\newblock {\em Biometrika}, 103(4):779--799.

\bibitem[Delaigle et~al., 2021]{delaigle2021estimating}
Delaigle, A., Hall, P., Huang, W., and Kneip, A. (2021).
\newblock Estimating the covariance of fragmented and other related types of
  functional data.
\newblock {\em J. Amer. Statist. Assoc.}, 116(535):1383--1401.

\bibitem[Descary and Panaretos, 2019]{descary2019recovering}
Descary, M.-H. and Panaretos, V.~M. (2019).
\newblock Recovering covariance from functional fragments.
\newblock {\em Biometrika}, 106(1):145--160.

\bibitem[Dudley, 2002]{dudley2002real}
Dudley, R.~M. (2002).
\newblock {\em Real analysis and probability}.
\newblock Cambridge University Press.

\bibitem[El\'ias et~al., 2023a]{elias2023integrated}
El\'ias, A., Jim\'enez, R., Paganoni, A.~M., and Sangalli, L.~M. (2023a).
\newblock Integrated depths for partially observed functional data.
\newblock {\em J. Comput. Graph. Statist.}, 32(2):341--352.

\bibitem[El\'ias et~al., 2023b]{elias2023depth}
El\'ias, A., Jim\'enez, R., and Shang, H.~L. (2023b).
\newblock Depth-based reconstruction method for incomplete functional data.
\newblock {\em Comput. Statist.}, 38(3):1507--1535.

\bibitem[El{\'\i}as and Nagy, 2024]{elias2024statistical}
El{\'\i}as, A. and Nagy, S. (2024).
\newblock Statistical properties of partially observed integrated functional
  depths.
\newblock {\em TEST}, pages 1--26.

\bibitem[Farewell et~al., 2022]{farewell2022missing}
Farewell, D.~M., Daniel, R.~M., and Seaman, S.~R. (2022).
\newblock Missing at random: a stochastic process perspective.
\newblock {\em Biometrika}, 109(1):227--241.

\bibitem[Ferraty and Vieu, 2006]{ferraty2006nonparametric}
Ferraty, F. and Vieu, P. (2006).
\newblock {\em Nonparametric functional data analysis}.
\newblock Springer Series in Statistics. Springer, New York.
\newblock Theory and practice.

\bibitem[Graf and Luschgy, 2000]{graf2000foundations}
Graf, S. and Luschgy, H. (2000).
\newblock {\em Foundations of quantization for probability distributions},
  volume 1730 of {\em Lecture Notes in Mathematics}.
\newblock Springer-Verlag, Berlin.

\bibitem[Hahn, 1978]{hahn1978central}
Hahn, M.~G. (1978).
\newblock Central limit theorems in {$D[0, 1]$}.
\newblock {\em Z. Wahrsch. Verw. Gebiete}, 44(2):89--101.

\bibitem[Hall and Van~Keilegom, 2007]{hall2007twosample}
Hall, P. and Van~Keilegom, I. (2007).
\newblock Two-sample tests in functional data analysis starting from discrete
  data.
\newblock {\em Statist. Sinica}, 17(4):1511--1531.

\bibitem[Horv{\'a}th and Kokoszka, 2012]{horvath2012inference}
Horv{\'a}th, L. and Kokoszka, P. (2012).
\newblock {\em Inference for functional data with applications}, volume 200.
\newblock Springer Science \& Business Media.

\bibitem[Hsing and Eubank, 2015]{hsing2015theoretical}
Hsing, T. and Eubank, R. (2015).
\newblock {\em Theoretical foundations of functional data analysis, with an
  introduction to linear operators}.
\newblock Wiley Series in Probability and Statistics. John Wiley \& Sons, Ltd.,
  Chichester.

\bibitem[Jiang and Arias-Castro, 2021]{jiang2021consistency}
Jiang, H. and Arias-Castro, E. (2021).
\newblock On the consistency of metric and non-metric k-medoids.
\newblock {\em International Conference on Artificial Intelligence and
  Statistics}, pages 2485--2493.

\bibitem[Kneip and Liebl, 2020]{kneip2020on}
Kneip, A. and Liebl, D. (2020).
\newblock On the optimal reconstruction of partially observed functional data.
\newblock {\em Ann. Statist.}, 48(3):1692--1717.

\bibitem[Kraus, 2015]{kraus2015components}
Kraus, D. (2015).
\newblock Components and completion of partially observed functional data.
\newblock {\em J. R. Stat. Soc. Ser. B. Stat. Methodol.}, 77(4):777--801.

\bibitem[Kraus, 2019]{kraus2019inferential}
Kraus, D. (2019).
\newblock Inferential procedures for partially observed functional data.
\newblock {\em J. Multivariate Anal.}, 173:583--603.

\bibitem[Kraus and Stefanucci, 2019]{kraus2019classification}
Kraus, D. and Stefanucci, M. (2019).
\newblock Classification of functional fragments by regularized linear
  classifiers with domain selection.
\newblock {\em Biometrika}, 106(1):161--180.

\bibitem[Kraus and Stefanucci, 2020]{kraus2020ridge}
Kraus, D. and Stefanucci, M. (2020).
\newblock Ridge reconstruction of partially observed functional data is
  asymptotically optimal.
\newblock {\em Statist. Probab. Lett.}, 165:108813, 5.

\bibitem[{Land Steiermark}, 2023]{data2023}
{Land Steiermark} (2023).
\newblock Air quality data.
\newblock \url{https://app.luis.steiermark.at/luft2/suche.php},.
\newblock Licensed under
  \href{https://creativecommons.org/licenses/by/4.0/deed.en}{CC BY 4.0};
  accessed on September 28, 2023.

\bibitem[Liebl, 2013]{liebl2013modeling}
Liebl, D. (2013).
\newblock Modeling and forecasting electricity spot prices: a functional data
  perspective.
\newblock {\em Ann. Appl. Stat.}, 7(3):1562--1592.

\bibitem[Liebl, 2019]{liebl2019nonparametric}
Liebl, D. (2019).
\newblock Nonparametric testing for differences in electricity prices: the case
  of the {F}ukushima nuclear accident.
\newblock {\em Ann. Appl. Stat.}, 13(2):1128--1146.

\bibitem[Liebl and Rameseder, 2019]{liebl2019partially}
Liebl, D. and Rameseder, S. (2019).
\newblock Partially observed functional data: the case of systematically
  missing parts.
\newblock {\em Comput. Statist. Data Anal.}, 131:104--115.

\bibitem[Liebl and Reimherr, 2023]{liebl2023fast}
Liebl, D. and Reimherr, M. (2023).
\newblock Fast and fair simultaneous confidence bands for functional
  parameters.
\newblock {\em J. R. Stat. Soc. Ser. B. Stat. Methodol.}, 85(3):842--868.

\bibitem[Lin and Wang, 2022]{lin2022mean}
Lin, Z. and Wang, J.-L. (2022).
\newblock Mean and covariance estimation for functional snippets.
\newblock {\em J. Amer. Statist. Assoc.}, 117(537):348--360.

\bibitem[Lin et~al., 2021]{lin2021basis}
Lin, Z., Wang, J.-L., and Zhong, Q. (2021).
\newblock Basis expansions for functional snippets.
\newblock {\em Biometrika}, 108(3):709--726.

\bibitem[Little, 1988]{little1988test}
Little, R. J.~A. (1988).
\newblock A test of missing completely at random for multivariate data with
  missing values.
\newblock {\em J. Amer. Statist. Assoc.}, 83(404):1198--1202.

\bibitem[Marczewski and Steinhaus, 1958]{marczewski1958on}
Marczewski, E. and Steinhaus, H. (1958).
\newblock On a certain distance of sets and the corresponding distance of
  functions.
\newblock {\em Colloquium Mathematicum}, 6(1):319--327.

\bibitem[Mealli and Rubin, 2015]{mealli2015clarifying}
Mealli, F. and Rubin, D.~B. (2015).
\newblock Clarifying missing at random and related definitions, and
  implications when coupled with exchangeability.
\newblock {\em Biometrika}, 102(4):995--1000.

\bibitem[Molchanov, 2017]{molchanov2017theory}
Molchanov, I. (2017).
\newblock {\em Theory of random sets}, volume~87 of {\em Probability Theory and
  Stochastic Modelling}.
\newblock Springer-Verlag, London, second edition.

\bibitem[Ofner and H\"ormann, 2025]{ofner2024covariateinformed}
Ofner, M. and H\"ormann, S. (2025).
\newblock Covariate-informed reconstruction of partially observed functional
  data via factor models.
\newblock {\em Electron. J. Stat.}, 19(1):1981--2000.

\bibitem[Pollard, 1981]{pollard1981strong}
Pollard, D. (1981).
\newblock Strong consistency of {$k$}-means clustering.
\newblock {\em Ann. Statist.}, 9(1):135--140.

\bibitem[{R Core Team}, 2025]{R}
{R Core Team} (2025).
\newblock {\em R: A Language and Environment for Statistical Computing}.
\newblock R Foundation for Statistical Computing, Vienna, Austria.

\bibitem[Ramsay and Silverman, 2005]{ramsay1997functional}
Ramsay, J.~O. and Silverman, B.~W. (2005).
\newblock {\em Functional data analysis}.
\newblock Springer Series in Statistics. Springer, New York, second edition.

\bibitem[Rubin, 1976]{rubin1976inference}
Rubin, D.~B. (1976).
\newblock Inference and missing data.
\newblock {\em Biometrika}, 63(3):581--592.

\bibitem[Seaman et~al., 2013]{seaman2013meant}
Seaman, S., Galati, J., Jackson, D., and Carlin, J. (2013).
\newblock What is meant by ``missing at random''?
\newblock {\em Statist. Sci.}, 28(2):257--268.

\bibitem[Zhang et~al., 2019]{zhangetal2019unified}
Zhang, S., Han, P., and Wu, C. (2019).
\newblock A unified empirical likelihood approach for testing {MCAR} and
  subsequent estimation.
\newblock {\em Scand. J. Stat.}, 46(1):272--288.

\end{thebibliography}


\begin{thebibliography}{}

\bibitem[Billingsley, 1995]{BillingsleyPatrick1979Pam}
Billingsley, P. (1995).
\newblock {\em Probability and measure}.
\newblock Wiley Series in Probability and Mathematical Statistics. John Wiley
  \& Sons, Inc., New York, third edition.
\newblock A Wiley-Interscience Publication.

\bibitem[Borodin and Salminen, 1996]{Borodin1996handbook}
Borodin, A. and Salminen, P. (1996).
\newblock {\em Handbook of Brownian Motion: Facts and Formulae}.
\newblock Birkh{\"a}user Verlag.

\bibitem[Bosq, 2000]{bosq2000linear}
Bosq, D. (2000).
\newblock {\em Linear processes in function spaces}, volume 149 of {\em Lecture
  Notes in Statistics}.
\newblock Springer-Verlag, New York.
\newblock Theory and applications.

\bibitem[Efron and Tibshirani, 1994]{efron1994introduction}
Efron, B. and Tibshirani, R.~J. (1994).
\newblock {\em An introduction to the bootstrap}.
\newblock Chapman and Hall/CRC.

\bibitem[Hahn, 1978]{hahn1978central}
Hahn, M.~G. (1978).
\newblock Central limit theorems in {$D[0, 1]$}.
\newblock {\em Z. Wahrsch. Verw. Gebiete}, 44(2):89--101.

\bibitem[Kraus, 2019]{kraus2019inferential}
Kraus, D. (2019).
\newblock Inferential procedures for partially observed functional data.
\newblock {\em J. Multivariate Anal.}, 173:583--603.

\bibitem[Liebl and Rameseder, 2019]{liebl2019partially}
Liebl, D. and Rameseder, S. (2019).
\newblock Partially observed functional data: the case of systematically
  missing parts.
\newblock {\em Comput. Statist. Data Anal.}, 131:104--115.

\bibitem[Paparoditis and Sapatinas, 2016]{paparoditis2016bootstrap}
Paparoditis, E. and Sapatinas, T. (2016).
\newblock Bootstrap-based testing of equality of mean functions or equality of
  covariance operators for functional data.
\newblock {\em Biometrika}, 103(3):727--733.

\bibitem[Rameseder and Liebl, 2017]{RamesederFD2017PartiallyFD}
Rameseder, S. and Liebl, D. (2017).
\newblock {\em PartiallyFD: Partially Observed Functional Data: The Case of
  Systematically Missings}.
\newblock \url{https://github.com/stefanrameseder/PartiallyFD}.

\bibitem[Rao, 1963]{rao1963law}
Rao, R.~R. (1963).
\newblock The law of large numbers for {$D[0, 1]$}-valued random variables.
\newblock {\em Theory of Probability \& Its Applications}, 8(1):70--74.

\bibitem[van~der Vaart, 1998]{Vaart1998Asymptotic}
van~der Vaart, A.~W. (1998).
\newblock {\em Asymptotic statistics}, volume~3 of {\em Cambridge Series in
  Statistical and Probabilistic Mathematics}.
\newblock Cambridge University Press, Cambridge.

\bibitem[Zhou and Pan, 2014]{zhou2014principal}
Zhou, L. and Pan, H. (2014).
\newblock Principal component analysis of two-dimensional functional data.
\newblock {\em J. Comput. Graph. Statist.}, 23(3):779--801.

\end{thebibliography}
\end{document}